\documentclass[12pt]{article}
\usepackage[top=1.25in,bottom=1in,left=1in,right=1in]{geometry}
\usepackage{graphicx}
\usepackage{amsmath}
\usepackage{bbm}
\usepackage{amsthm}
\usepackage{amsfonts}
\usepackage{wrapfig}
\usepackage{multicol}
\usepackage{float}
\usepackage{pdfpages}
\usepackage{mathabx}
\usepackage{ amssymb }
\usepackage{setspace}
\newcommand{\RN}[1]{%
  \textup{\uppercase\expandafter{\romannumeral#1}}%
}

\usepackage[colorlinks=true, urlcolor=blue, linkcolor=blue,citecolor=blue]{hyperref}

\usepackage{natbib}

\providecommand{\bz}{\bold{z}}

\newtheorem{claim}{Claim}

\newtheorem{proposition}{Proposition}
\newtheorem{theorem}{Theorem}
\newtheorem{corollary}{Corollary}
\newtheorem{lemma}{Lemma}

\newtheorem{assumption}{Assumption}

\newtheorem{observation}{Observation}
\newtheorem{example}{Example}
\usepackage{mathrsfs}
\usepackage{natbib}
\theoremstyle{definition}
\newtheorem{remark}{Remark}
\newtheorem{definition}{Definition}

\DeclareMathOperator*{\argmax}{argmax}
\DeclareMathOperator*{\argmin}{argmin}
\DeclareMathOperator*{\supp}{supp}
\DeclareMathOperator*{\marg}{marg}

\title{Games of Incomplete Information\\ Played By Statisticians} 
\author{Annie Liang\thanks{Department of Economics, University of Pennsylvania. Email: anliang@upenn.edu. I am especially grateful to Drew Fudenberg for his guidance on this paper. This paper also benefitted from useful comments and suggestions by Jetlir Duraj, Siddharth George, Ben Golub, Jerry Green, Philippe Jehiel, Scott Kominers, David Laibson, Jonathan Libgober, Erik Madsen, Stephen Morris, Sendhil Mullainathan, Mariann Ollar, Harry Pei, Andrei Shleifer, Dov Samet, Tomasz Strzalecki, Satoru Takahashi, Anton Tsoy, and Muhamet Yildiz.}}

\begin{document}
%
%
%
%
%
%
%






\pagenumbering{arabic}
\maketitle


\linespread{1.2}

\begin{abstract}


Players are statistical learners who learn about payoffs from data. They may interpret the same data differently, but have common knowledge of a class of learning procedures. I propose a metric for the analyst's ``confidence" in a strategic prediction, based on the probability that the prediction is consistent with the realized data. The main results characterize the analyst's confidence in a given prediction as the quantity of data grows large, and provide bounds for small datasets. The approach generates new predictions, e.g. that speculative trade is more likely given high-dimensional data, and that coordination is less likely given noisy data.


\end{abstract}


\linespread{1.4}

\section{Introduction}

Predictions of play in incomplete information games depend crucially on the beliefs of the agents, but we rarely know what those beliefs are. A standard approach to modeling beliefs assumes that players share a common prior belief over states of the world, and form posterior beliefs using Bayesian updating. Under this approach, posterior beliefs that are commonly known must be identical \citep{disagree}, and repeated communication of beliefs eventually leads to agreement \citep{geanakoplos}. These implications conflict not only with considerable empirical evidence of public and persistent disagreement, but also with the more basic experience that individuals interpret the same information in different ways.\footnote{For example, in financial markets, individuals publicly disagree in their interpretations of earnings announcements \citep{kandel}, valuations of financial assets \citep{complexasset}, and forecasts for inflation \citep{inflation}.} 

This paper relaxes the assumption of a common prior by supposing that players are statistical learners: they form beliefs about payoff-relevant parameters based on data, but potentially disagree on how to interpret that data. I define a \emph{learning rule} to be any function  that maps data (a sequence of signals) into a belief distribution over payoff-relevant parameters (a first-order belief). Players have common knowledge of some set of reasonable learning rules---for example, these learning rules may correspond to Bayesian updating from a set of prior beliefs, or they may be maps from data to beliefs based on frequentist estimates for the unknown parameter. The special case of a singleton Bayesian learning rule returns the common prior assumption, but in general, the set of learning rules will produce different beliefs from the same data, which I interpret as the set of \emph{plausible beliefs}. I impose a key restriction to structure the approach: for any realization of the data, each player's own belief about the parameter is a plausible belief; they assign probability 1 to all other players holding plausible beliefs; and so forth. Since the set of plausible beliefs is endogenous to the (random) data, so too are the strategic predictions that are consistent with this belief restriction.

The main contribution of the paper is a proposed metric for the analyst's ``confidence" in a strategic prediction in this game. Specifically, consider the prediction that a given action is rationalizable. I quantify the analyst's confidence in this prediction via a confidence set: The upper bound of the confidence set is the probability that the action is rationalizable given \emph{some}  belief satisfying the belief restriction, and the lower bound is the probability that the prediction holds for \emph{all} beliefs satisfying the restriction. Thus, if both of these probabilities are equal to one, the analyst has maximal certainty that the action is rationalizable, and if they are both zero, he has maximal certainty that it is not.  In the intermediate cases, there is uncertainty about whether the action is rationalizable, and the confidence sets present a way to quantify the extent of that uncertainty.

The main results in this paper characterize various properties of these confidence sets, beginning with their asymptotic behaviors as the quantity of data grows large. I first show that if sets of learning rules are too large, then the confidence set may fail to be continuous at infinite data: that is, even if an action is \emph{strictly} rationalizable in the limiting infinite-data game, the analyst's confidence sets may be very different from $\{1\}$ for arbitrarily large quantities of data. Roughly, this is because the rate of convergence under different learning rules cannot be uniformly bounded, so it is always possible that some learning rule produces a belief that is very different from the others. If, however, the set of learning rules satisfy a uniform convergence property that I describe, then the following statements hold: If an action is strictly rationalizable at the limit, then the analyst's confidence set must converge to $\{1\}$ as the quantity of data gets large, and if an action is \emph{not} rationalizable in the limit, the analyst's confidence set converges to $\{0\}$. (The intermediate case, in which actions are rationalizable but not strictly rationalizable, is more subtle---see Section \ref{sec:Asymptotic} for further detail.)  

Next, I consider the setting of small sample sizes, and bound the extent to which the analyst's confidence set differs from its asymptotic limit. These bounds depend on properties of the learning environment---specifically, the quantity of data, and how fast the different learning rules jointly recover the payoff-relevant parameter---as well as on a cardinal measure for how strict the solution is at the limit. I apply these bounds to characterize confidence sets for example games and sets of learning rules. They allow us to obtain specific, quantitative, statements about confidence away from the limit of infinite data.

In some cases it is possible to go beyond these bounds, and to fully characterize the confidence set. Two such examples given in Section \ref{sec:example}, where I consider a trade game and a coordination game. I show that the proposed approach generates novel comparative statics in these games: Speculative trade is predicted to be more plausible when agents learn from higher-dimensional data, and coordination is predicted to be less likely when agents observe noisy data. These predictions---which hold even though data is public and common---are difficult to produce under assumption of a common prior. 

\subsection{Related Literature}
  
This paper builds on a literature on the role of the common prior assumption in economic theory. (See \citet{cpa} for a survey of key conceptual points.) Here I focus on an argument that even if learning does produce common priors in the long run, this does not imply that we should see common priors given a finite quantity of data, especially if that data is complex and hard to interpret. Rather than taking the limiting common prior as one that is already reached, I ask what predictions we can make while data is still being accumulated. The confidence sets introduced in this paper provide a quantitative account of whether predictions implied by a (limiting) common prior also hold for small data sets.\footnote{Other reasons that the common prior is tenuous include that the data itself may lead to incomplete learning if it is endogenously acquired, and that convergence of individual beliefs need not imply convergence in beliefs about beliefs \citep{commonlearning,yildiz}.}

This paper  also contributes to a literature on the robustness of strategic predictions to the specification of player beliefs \citep{emailgame,strategic,WY,faingold,criticaltypes} and equilibrium selection in incomplete information games \citep{globalgames,kajiimorris}. At a technical level, the actions that are rationalizable under the belief restriction that I impose are $\Delta$-rationalizable strategies of \citet{battigalli}, where the set $\Delta$ of first-order belief restrictions is endogenous to a learning process.\footnote{Such belief restrictions have also been usefully applied in the work of \citet{OllarPenta}, among others.} The permitted types converge in the \emph{uniform-weak topology}, as proposed and characterized in \citet*{faingold} and \citet*{faingold2}, and I use results about this topology to prove several of the main results.

Conceptually, the goals of the present paper differ from the previous literature in several respects: First, my focus here is not on equilibrium selection---choosing one equilibrium from a set of many---but rather on providing a metric for confidence in a given prediction. Second, in contrast to the many binary or ``qualitative" notions of robustness that have been proposed, 
this paper delivers a quantitative metric. Third, while the literature has primarily considered robustness to \emph{perturbations} of beliefs, I am interested here also in predictions that we may make for beliefs that are far from the limiting beliefs.   To discipline these beliefs, I endogenize the type space using a statistical learning foundation for belief formation. This aspect of the paper---combining learning foundations with game theoretic implications---connects to papers such as  \citet{learningbayesian}, \citet{esponda}, and \citet{stewart}, among others.\footnote{\citet{BrandenburgerFriedenbergKeisler} and \citet{BattigalliPrestipino} also motivate small type structures as emerging from learning,  although they do not explicitly model a dynamic learning process.} The modeling of agents as ``statisticians" or ``machine learners" relates to a growing literature in decision theory \citep{gilboa3,gilboa4,statisticians,pai} and game theory \citep{analogies,spiegler,competingmodels,CherrySalant,Jackson}. Of these papers, \citet{stewart} and \cite{CherrySalant} are closest: \citet{stewart} characterizes the limiting equilibria of a sequence of games in which players infer payoffs from related games, and \citet{CherrySalant} models players as statisticians who form beliefs about the action distribution  based on statistical inference from a sample of observed players. My goal here differs in that it is to provide a metric of robustness rather than a new solution concept.

 \section{Examples} \label{sec:example}
In this section, I use the proposed approach to revisit some classic examples\textemdash  a two-player coordination game and a two-player trading game. In each of these games, the assumption that players share a common prior has strong implications for strategic play. I show how we can relax this assumption by endogenizing disagreement based on two new primitives: a data-generating process and set of learning rules. This approach allows us to relate confidence in a strategic prediction to primitives of the learning environment.

\subsection{Trade} \label{sec:ExampleTrade}
A Seller owns a good of unknown value $v \in \{0,1\}$. He can either \emph{enter} a market at cost $c$, or \emph{exit} and keep the good. Entering leads to a simultaneous interaction with a Buyer, where the Seller chooses whether to sell the good at a (pre-set) posted price $p$, and the Buyer chooses whether to purchase the good at that price. The game and its payoffs are described in Figure \ref{fig:game} below.

\begin{figure}[h]
\centering
\includegraphics[scale=0.7]{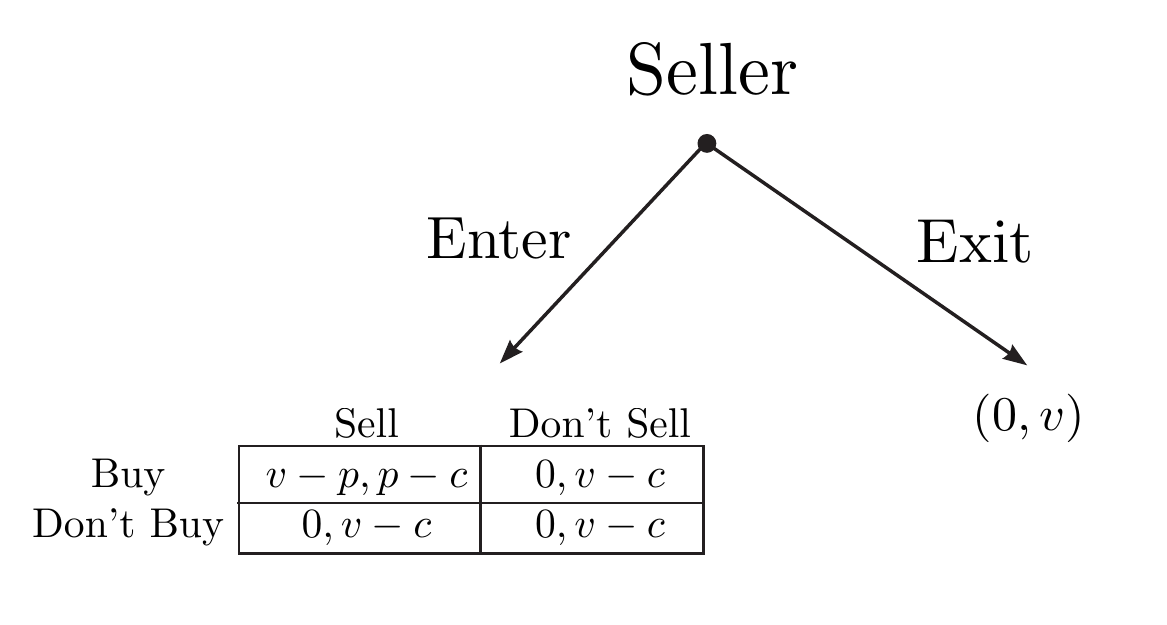}
\caption{Description of Game} \label{fig:game}
\end{figure}
Suppose that the cost $c$ and price $p$ satisfy $0<c<p<1$, so the Seller prefers to sell at the low value and prefers to keep the good at the high value. If players share a common prior about $v$, then entering is not rationalizable for the Seller in this game, so trade will not occur (similar to the no-trade theorem of \citet{notrade}).\footnote{If trade does not occur subsequently, then the Seller receives $v-c$ from entering but $v>v-c$ from exiting. Thus,  entering can be rationalized only if trade subsequently occurs. But trade can occur only if the Buyer  believes that  $\mathbb{E}(v)\geq p$ while the Seller believes that $\mathbb{E}(v) \leq p$, implying $\mathbb{E}(v)=p$ under their shared belief. The Seller can improve on his expected payoff of $p-c$ by choosing to exit. } 

I suppose instead that player form beliefs based on a common data set of past goods and their valuations, but draw inferences from this data in different ways. Each good in the data is described by $m$ observable attributes with values  normalized to lie in the interval $[-1,1]$. Typical attributes are denoted $x\in X:=[-1,1]^m$, and the value of a good is a deterministic function $f: X \rightarrow \{0,1\}$ of its attributes. The public data $\bz_n=\{(x_i,f(x_i))\}_{i=1}^n$ is a sequence of $n$ goods with attributes $x_i$ drawn from a uniform distribution on $X$,  and the values of these goods.

Players have common knowledge that $f$ belongs to a certain family of functions $\mathcal{F}$, which we can think of as the relevant set of models.  For simplicity, let $\mathcal{F}$ be the set of \emph{rectangular classification rules}, i.e. functions $f_R(x) = \mathbbm{1}(x\in R)$ indexed to hyper-rectangles $R$ in $[-1,1]^m$.\footnote{For example, whether all attributes fall into an ``acceptable" range, as judged by a downstream buyer.} The attributes of the Seller's good are known to be the zero-vector, denoted $x^S$, so the agent's beliefs over $\mathcal{F}$ determine their beliefs about $v=f(x^S)$, the value of the Seller's good.\footnote{In more detail, the state space is $\Omega = \mathcal{F} \times (X \times \{0,1\})^\infty$ and the payoff-relevant parameter is $f(x^S)$. Agents have a prior belief $q$ over $\Omega$, where $\marg_\mathcal{F} q$ has support on the set of models $R$. Conditional on the true model being $f$, the data-generating process over $ (X \times \{0,1\})^\infty$ is $q( (x,y) \mid f ) = g(x) \cdot \mathbbm{1}(y=f(x))$ where $g$ is the uniform density on $X$.}

Agents may have different prior beliefs $\pi \in \Delta(\mathcal{F})$ over the set of models. Fixing a prior $\pi$, the posterior belief $\pi(f \mid \bz_n)$ given data $\bz_n$ is a re-normalization of the prior over all rules consistent with the observed data (see Figure \ref{fig:rectangle}), and the posterior probability assigned to the Seller's good having a high value is $\pi(\{f: f(x^S) = 1\} \mid \bz_n)$. Define $\mathscr{B}(\bz_n) \subseteq \Delta(\{0,1\})$ to be the set of all posterior beliefs about $v$ that are consistent with Bayesian updating from some prior $\pi \in \Delta(\mathcal{F})$ and the data $\bz_n$.

\begin{figure}[h]
\centering
\includegraphics[scale=0.5]{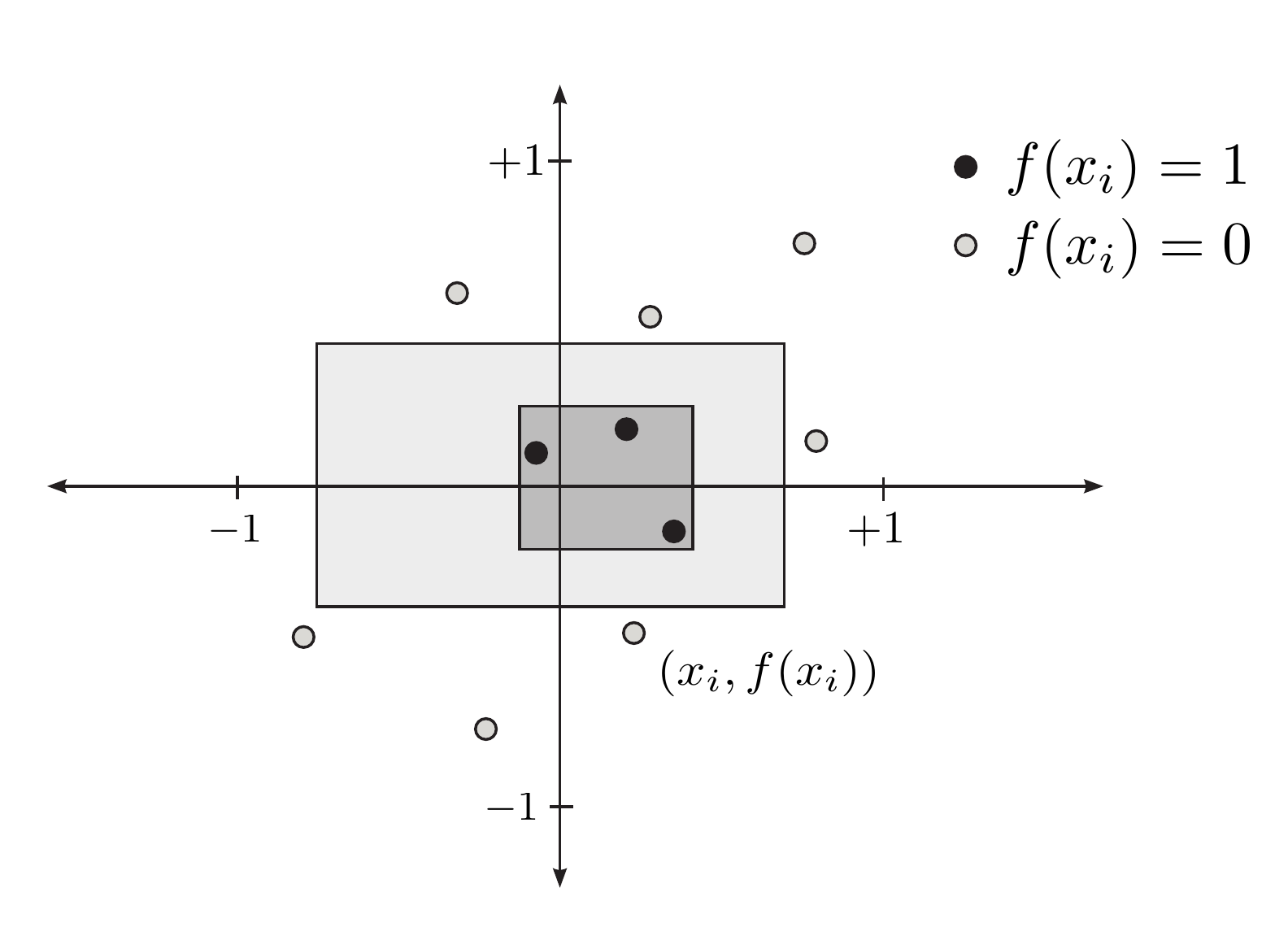}
\caption{\footnotesize{The circles represent the observed data. Each good is described by a vector in $[-1,1]\times[-1,1]$. The circle is black if its valuation is 1 and gray if its valuation is low. A rule is consistent with the data if it correctly predicts the valuation for each observation. Two rectangular classification rules are depicted: each predicts `1' for goods in the shaded region and `0' for goods outside. Both are consistent with the observed data.}} \label{fig:rectangle}
\end{figure}

Without making detailed assumptions about the priors that players hold, or the beliefs that players have about the priors of other players, I place the following restriction.
  
\begin{assumption}[Restriction on Beliefs] \label{ass:Example} For every realized $\bz_n$, players have common certainty in the event that all players have a first-order belief in 
$\mathscr{B}(\bz_n)$.\footnote{See Section \ref{subsec:restrictBelief} for the precise definition of common certainty in such an event.} \end{assumption}

\noindent That is, players are assumed to have beliefs consistent with Bayesian updating from some prior on $\mathcal{F}$; they assign probability 1 to the other player having a posterior belief in this set, and so on.

The analyst's (ex-ante) \emph{confidence} in predicting that entering is rationalizable for the Seller is defined as follows: For every number of observations $n$, let $\overline{p}^n$ be the measure of data sets $\bz_n$ such that entering is rationalizable for \emph{some} belief satisfying Assumption \ref{ass:Example}, and let  $\underline{p}^n$ be the measure of data sets $\bz_n$ such that entering is rationalizable for \emph{all} beliefs satisfying Assumption \ref{ass:Example}. Then the set $[\underline{p}^n, \overline{p}^n]$ represents a ``confidence set" that describes how certain the analyst should be in predicting that the entering is rationalizable, when agents observe $n$ common data points. The extreme case of $\overline{p}^n=\underline{p}^n=1$ implies full confidence that entering is rationalizable, while $\overline{p}^n=\underline{p}^n=0$ implies full confidence that it is not.

\begin{claim} The probability $\underline{p}^n=0$ for every $n\in \mathbb{Z}_+$. Additionally:
\begin{itemize}
\item[(a)] Fixing any number of attributes $m\in \mathbb{Z}_+$, the probability $\overline{p}^n \rightarrow 0$ as  $n \rightarrow \infty$.
\item[(b)] Fixing any number of data observations $n\in \mathbb{Z}_+$, the probability $\overline{p}^n$ is  increasing in the number of attributes $m$, and $\overline{p}^n \rightarrow 1$ as $m \rightarrow \infty$. 
\end{itemize} \label{claim:Trade}
\end{claim}

That is, for every quantity of data $n$, the probability that entering is rationalizable for all beliefs satisfying Assumption \ref{ass:Example} is zero. But the probability that entering is rationalizable for some belief satisfying Assumption \ref{ass:Example} varies depending on $n$ and $m$. Part (a) says that as the number of observations $n$ grows large, the probability $\overline{p}^n$ vanishes to zero, implying that the confidence set converges to a degenerate interval at zero. The infinite-data limit thus returns the prediction of ``no trade" consistent with the common prior assumption. But if the quantity of data is finite, and the number of attributes is large, then $\overline{p}^n$ can be substantially greater than zero. Indeed, Part (b) of the claim says that this probability $\overline{p}^n$ can be made arbitrarily close to 1 by increasing the dimensionality of the learning problem via choice of large $m$. This reflects that in a high-dimensional learning problem, many classification rules are likely to be consistent with the data, including some that yield conflicting predictions. Thus, ``rational" disagreement is possible and even likely.

An exact characterization of $[\underline{p}^n,\overline{p}^n]$ is given in Lemma \ref{lemm:pnRectangle} in the appendix. Using this lemma, I plot in Figure \ref{fig:Trade} the behavior of these confidence sets for different values of $m$, assuming the true function to be $\mathbbm{1}(x \in [-0.1,0.1]^m)$. 
Although trade is not predicted in the limiting game, it is a plausible outcome if the number of data observations is small and the number of attributes is large. As one specific example: If there are 10 attributes and players observe only 20 goods, then the confidence set $[\underline{p}^n,\overline{p}^n]=[0,0.99]$. That is, with near certainty, entering will be rationalizable for the Seller given the realized data for some belief satisfying Assumption \ref{ass:Example}. While I chose a simple set of learning rules for the purpose of obtaining exact expressions for the confidence set, in practice confidence sets like those depicted in Figure \ref{fig:Trade} can be simulated for more complex kinds of learning procedures.

\begin{figure}[h] 
\centering
\includegraphics[scale=0.40]{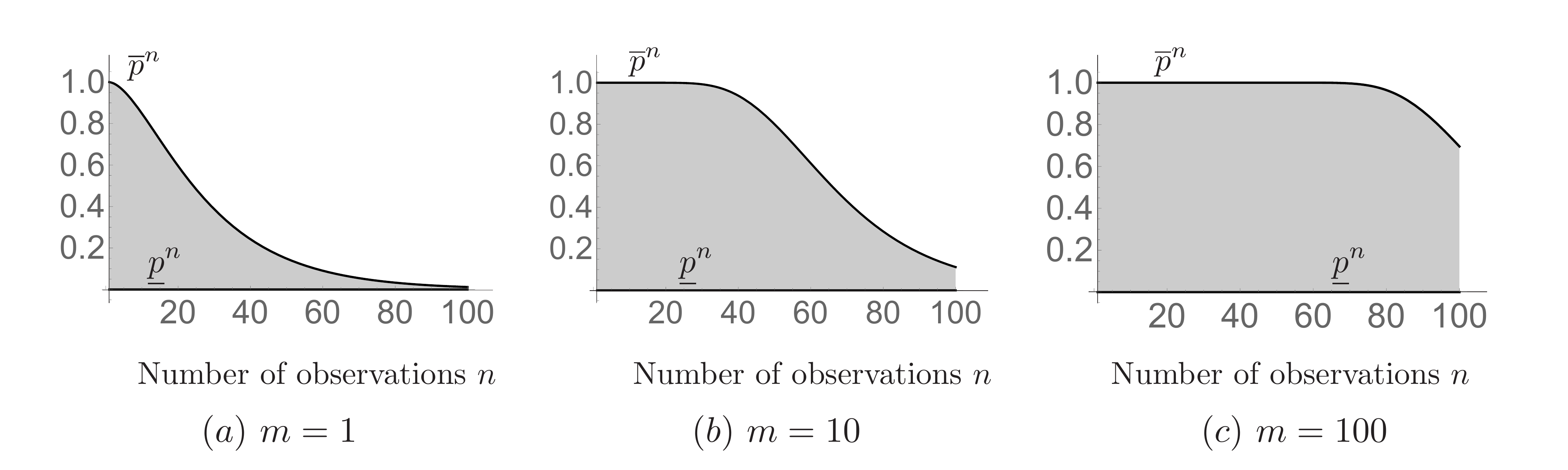}
\caption{The shaded area depicts confidence sets $[\underline{p}^n,\overline{p}^n]$ for the rationalizability of entering given $n$ common observations.} \label{fig:Trade}
\end{figure}

\subsection{Coordination} \label{sec:ExampleCoord}

I now turn to a second classic game, where (in contrast to the previous example) the prediction of interest holds in the infinite-data limit. I characterize the robustness of that prediction when players have beliefs based on small quantities of data.

A contagious disease spreads across a population at an unknown speed. Two states are connected by travel, and their governors choose between implementing a \emph{strong} or a \emph{weak} lockdown policy in their states to slow the spread of the disease. Implementation of the strong lockdown policy entails a large economic cost, but if the states coordinate on this policy, then the disease will be suppressed with certainty.

The two governors form beliefs about the growth rate of the disease based on a public data set $\{(t, y_t)\}_{t=1}^n$, which consists of the number of reported cases of the disease, $y_t$, on days $t=1, 2, \dots, n$. The number of reported cases grows exponentially according to 
\[\log y_t = \beta t  + \varepsilon_t\]
where the noise term $\varepsilon_t$ has a normal distribution with known parameters $\mu=0$ and $\sigma^2>0$. The constant $\beta$ is not known. Payoffs are given by the following matrix:
\[\begin{array}{ccc}
& \emph{Strong} & \emph{Weak} \\
\emph{Strong} & -1,-1 & -1-\beta, -\beta \\
\emph{Weak} & -\beta, -1-\beta & -\beta,-\beta
\end{array}\]
The economic cost of the strong lockdown is normalized to 1, and the cost of letting the disease progress without a strong lockdown is given by its growth rate $\beta$. A weak lockdown is strictly dominant if $\beta<1$, but coordination on the strong lockdown is the Pareto-dominant Nash equilibrium if $\beta>1$. 

Since the data is public, the assumption of a common prior over the parameter $\beta$  restricts  agents to identical posterior beliefs. This rules out, for example, the possibility that a governor chooses the weak policy, not because he infers from the data that the disease is low-risk, but because he believes that the other governor may make such an inference.

In practice, agents use statistical procedures to infer $\beta$ from the data. I thus relax the assumption of a common prior in the following way. Define $\hat{\beta}(\bold{z}_n)$ to be the ordinary least-squares estimate of $\beta$ from the data $\bold{z}_n:=\{(t, \log y_t)\}_{t=1}^n$, and let $\phi_n$ be the constant such that $\mathscr{C}(\bold{z}_n)=\{\beta: \vert \beta - \hat{\beta}(\bz_n) \vert \leq \phi_n\}$ is a $(1-\alpha)$-confidence interval for $\beta$, where $\alpha\in (0,1)$ is fixed. Suppose the governors have common certainty in the event that both of their first-order beliefs have support on $\mathscr{C}(\bz_n)$; that is, each governor's own belief about the value of $\beta$ has support in this confidence interval; they assign probability 1 to the other agent having such a belief; and so forth. This allows players to disagree given the data, but requires the size of that disagreement to be confined within the confidence interval. (The assumption that players assign common \emph{certainty} is not crucial. Similar results go through if the players have common $p$-belief \citep{monderersamet} in the set for large $p$, so long as we bound the parameter space for $\beta$. See Section \ref{sec:extensions}.) 

\begin{claim} \label{claim:Coord} Suppose that the actual growth rate is fast $(\beta>1$), so that the strong lockdown is rationalizable given complete information of the payoffs. Then, for every $\sigma^2>0$, both $\overline{p}^n$ and $\underline{p}^n$ are increasing in $n$, while for every $n$, both $\overline{p}^n$ and $\underline{p}^n$ are decreasing in $\sigma^2$.
\end{claim}

 That is, the analyst gains confidence in predicting that the strong lockdown is rationalizable as the reporting noise $\sigma^2$ decreases, and the number of observations $n$ increases.\footnote{If instead $\beta<1$, then the reverse statements hold; that is, the probabilities $\overline{p}^n$ and $\underline{p}^n$ are decreasing in $n$ and increasing in $\sigma^2$.} Lemma \ref{lemma:exCoord} in the appendix explicitly characterizes $\underline{p}^n$ and $\overline{p}^n$. Using the expressions in this lemma, I plot in Figure \ref{fig:Coord} the behavior of these confidence sets for different levels of reporting noise $\sigma$.

\begin{figure}[h]
\centering
\includegraphics[scale=0.39]{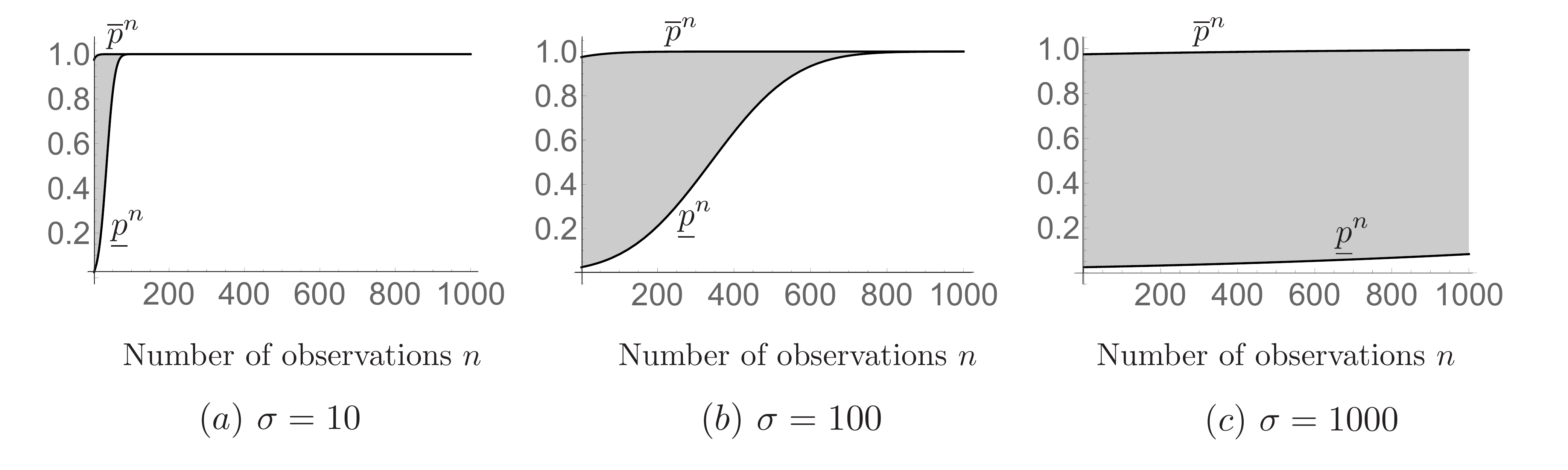}
\caption{The shaded area depicts confidence sets $[\underline{p}^n,\overline{p}^n]$ for the rationalizability of the strong lockdown given $n$ common observations, and allowing the reporting noise $\sigma$ to vary. In all panels, $\beta=2$ and $\alpha=0.05$.} \label{fig:Coord}
\end{figure}

As the number of observations $n$ grows, both $\underline{p}^n$ and $\overline{p}^n$ increase as well. If the number of observations is large relative to the reporting noise, then the analyst should have high confidence that the strong lockdown policy is rationalizable---for example, if $\sigma=10$ and $n=100$, then the confidence set $[0.99,1]$ is nearly degenerate at certainty. On the other hand, if the reporting noise is large and the number of observations is relatively small, then the strong lockdown is likely to be rationalizable for some permitted types, but \emph{not} for all of them---for example, if $\sigma=n=100$, then the confidence set is $[0.08,0.99]$, suggesting substantial ambiguity regarding whether the strong lockdown is a good prediction of play.

Subsequently, I generalize the approach described in these two examples. 

\section{Approach} \label{sec:Setting}

\subsection{Preliminaries} \label{sec:preliminaries}

\paragraph{Basic Game.}

There is a finite set $\mathcal{I}$ of players and a finite set of actions $A_i$ for each player $i$. The set of action profiles is $A=\prod_{i \in \mathcal{I}} A_i$, and the set of possible games is identified with $U:=\mathbb{R}^{|\mathcal{I}| \times |A|}$. Agents have beliefs over a set of payoff-relevant parameters $\Theta$, which is a compact and convex subset of finite-dimensional Euclidean space. It is possible to take $\Theta$ to be a subset of $U$, so that each $\theta$ is itself a game, or to define beliefs over a lower-dimensional set of payoff-relevant parameters as in Section \ref{sec:example}. In either case, the parameters in $\Theta$ are assumed to be related to payoffs by a bounded and Lipschitz continuous embedding $g: \Theta \rightarrow U$ (assuming the sup-norm on both spaces).\footnote{The map $g$ can be interpreted as capturing the known information about the structure of payoffs.}

\paragraph{Beliefs.} For each player $i$, let $X^0_i = \Theta$, $X_i^1 = X_i^0 \times \prod_{j\neq i} \Delta(X_j^0)$, $\dots$, 
$X_i^n = X_i^{n-1} \times \prod_{j \neq i} \Delta(X_j^{n-1})$, etc., 
so that each $X_i^k$ is the set of possible $k$-th order beliefs for player $i$. Define $T_i^0 = \prod_{n=0}^\infty \Delta(X^i_n)$. An element $(t^1_i, t^2_i, \dots) \in T_i^0$ is a \emph{hierarchy of beliefs} over $\Theta$ (describing the player's uncertainty over $\Theta$, his uncertainty over his opponents' uncertainty over $\Theta$, and so forth), and referred to simply as a \emph{belief} or \emph{type}.\footnote{Types are sometimes modeled as encompassing all uncertainty in the game. In the present paper, types describe players' structural uncertainty over payoffs, but not their strategic uncertainty over opponent actions.} There is a subset of types $T^*_i$ (that satisfy the property of \emph{coherency}\footnote{$\marg_{X_{n-2}} t_i^n = t_i^{n-1}$, so that $(t_i^1, t_i^2, \dots)$ is a consistent stochastic process.} and common knowledge of coherency) and a function $\kappa^*_i: T^*_i \rightarrow \Delta \left(\Theta \times T^*_{-i} \right)$ such that $\kappa_i^*(t_i)$ preserves the beliefs in $t_i$; that is, $\marg_{X_{n-1}}\kappa^*_i(t_i) = t^n_i$ for every $n$ \citep{mertens,typespace}.\footnote{The notation $T^*_{-i}$ is used throughout the paper to denote the set of profiles of opponent types, $\prod_{j \neq i} T^*_j$.} The tuple $(T^*_i, \kappa^*_i)_{i \in \mathcal{I}}$ is the \emph{universal type space}. Subsequently I will develop smaller type spaces $(T_i,\kappa_i)_{i \in \mathcal{I}}$ where each $T_i\subseteq T_i^*$ and $\kappa_i: T_i \rightarrow \kappa_i^*(T_i)$ is the restriction of $\kappa_i^*$ to $T_i$.

\subsection{Restriction on Beliefs} \label{subsec:restrictBelief}

The proposed approach endogenizes the type space based on two new primitives: a \emph{data-generating process}, and a \emph{set of rules} for how to extrapolate beliefs from realized data.

Formally, let $(Z_t)_{t\in \mathbb{Z}_+}$ be a stochastic process where the random variables $Z_t$ take value in a common set $\mathcal{Z}$, and the typical sample path is denoted $\bold{z}=(z_1, z_2, \dots)$. The \emph{data-generating process} is a measure $P$ over the set $\mathcal{Z}^\infty$ of all (infinite) sample paths. Let $P^n$ denote the induced measure on the first $n$ variables.
 A \emph{data set} $\bold{z}_n$ of size $n$ is the restriction of $\bold{z}$ to its first $n$ coordinates, and $\mathcal{Z}^n$ is the set of all length-$n$ data sets. I use $Z^n=(Z_1, \dots, Z_n)$ to denote the random initial sequence of length $n$.\footnote{Throughout, symbols such as $Z_t$ and $Z^n$ denote random variables, whereas lowercase symbols such as $\bz_n$ are particular, constant values.}
 
A \emph{learning rule} is defined to be any map from data sets into first-order beliefs:
\[\mu: \bigcup_{n=1}^\infty \mathcal{Z}^n \rightarrow \Delta(\Theta).\]
The special case of Bayesian learning rule is identified with a distribution on $\Theta$ along with a family of distributions $(P_\theta)_{\theta \in \Theta}$, where each $P_\theta \in \Delta(\mathcal{Z}^\infty)$ is the data-generating distribution given parameter $\theta$.  The realized data $\bz_n$ determines a posterior distribution on $\Theta \times \mathcal{Z}^\infty$, and the learning rule maps this data into the marginal posterior distribution on $\Theta$.\footnote{In general, making the mapping deterministic may require choosing a conditional probability if there are multiple ones consistent with Bayes' rule; here and elsewhere in the paper, implicitly assume that the updating rule specifies such a choice when Bayesian rules are mentioned.}  Other learning rules may map the data $\bold{z}_n$ to a degenerate belief at a sample statistic (such as the empirical average) or to a distribution over various point-estimates for $\theta$. 

Players have common knowledge of a set $\mathcal{M}$ of learning rules.\footnote{Note that this set is common across players.} The set of \emph{plausible beliefs} given this set $\mathcal{M}$ and realized data set $\bz_n$ is defined 
\begin{equation}
\label{def:plausibleTheta}
\mathscr{B}(\bz_n) = \{\mu(\bz_n) \,: \,  \mu \in \mathcal{M}\} \subseteq \Delta(\Theta).
\end{equation}
Throughout this paper, I restrict consideration to sets of learning rules $\mathcal{M}$ that satisfy the following condition:
\begin{assumption}[Common Limiting Belief] \label{ass:limitbelief} There is a limiting belief $\mu^\infty$ such that
\[\lim_{n\rightarrow \infty} d_P(\mu(Z^n), \mu^\infty)  \rightarrow 0 \quad \mbox{$P$-a.s.} \quad \forall \mu \in \mathcal{M}\]
where $d_P$ is the Prokhorov metric on $\Theta$.\footnote{For any $\nu,\nu' \in \Delta(\Theta)$, the Prokhorov distance between these measures is $d_P(\nu,\nu') = \inf \{ \epsilon>0 \, : \, \nu(A) \leq \nu'(A^\epsilon) + \epsilon \mbox{ for all Borel-measurable } A \subseteq \Theta\}$, where $A^\epsilon$ denotes the $\epsilon$-neighborhood of $A$ in the sup-norm.} 
\end{assumption}
\noindent This assumption requires that all learning rules in $\mathcal{M}$ return the same limiting belief $\mu^\infty$ as the quantity of data $n$ grows large.\footnote{This limiting belief $\mu^\infty$ can be interpreted as a common prior, following what \citet{cpa} calls the ``frequentist justification" for assumption of a common prior.} It is not critical that all differences in beliefs are removed in the limit (see Section \ref{sec:extensions}). But maintaining Assumption \ref{ass:limitbelief} in the main text allows us to explore more precisely the scope for disagreement in an environment in which learning is feasible, but not immediate. While learning has not ceased, the available information may permit multiple different interpretations, producing differences in beliefs.

 I use the sets of plausible beliefs $\mathscr{B}(\bz_n)$  to impose a restriction on hierarchies of beliefs, as in \citet{battigalli}. Specifically, players are assumed to have \emph{common certainty} in the event that all players have first-order beliefs in $\mathscr{B}(\bz_n)$---that is, they have a first-order belief in this set, believe with probability 1 that all other players have a first-order belief in this set, and so forth \citep{monderersamet}.\footnote{As I discuss in Section \ref{sec:extensions}, it is not critical that players have common \emph{certainty} in this event, and the restriction can be relaxed to common $p$-belief for large $p$.} 

Formally, for any set $\mathscr{B} \subseteq \Delta(\Theta)$, and for any player $i$, define 
$$B^{1,1}_i(\mathscr{B}) := \left\{t_i \in T^*_i \, : \,  \textstyle \marg_{\Theta} \kappa^*_i(t_i) \in \mathscr{B} \right\}$$
to be the set of player $i$ types whose marginal beliefs over $\Theta$ belong to the set $\mathscr{B}$. For each $k>1$, and again for each player $i$, recursively define\footnote{
For example, $B^{2,1}_i(\mathscr{B})$ is the set of player $i$ types that assign probability 1 to all other players assigning probability 1 to $\mathscr{B}$.}
\[B^{k,1}_i(\mathscr{B})= \left\{t_i \in T^*_i \, : \, \kappa^*_i(t_i)\left(\Theta \times \prod_{j \neq i} B_j^{k-1,1}(\mathscr{B})\right) = 1\right\}.\] Then finally,
$T_i^{\mathscr{B}} =  \bigcap_{k \geq 1} B_i^{k,1}(\mathscr{B})$ 
is the set of player $i$ types that have common certainty in the event that all players' first-order beliefs belong to $\mathscr{B}$. 

\begin{definition} \label{def:typespace} 
For every $\bz_n$, the induced type space is 
$\left(T^{\mathscr{B}(\bz_n)}_i, \kappa^{\mathscr{B}(\bz_n)}_i\right)_{i \in \mathcal{I}}$, where $\kappa^{\mathscr{B}(\bz_n)}_i: T_i^{\mathscr{B}(\bz_n)} \rightarrow \kappa_i^*\left(T_i^{\mathscr{B}(\bz_n)}\right)$ is the restriction of $\kappa_i^*$ to $T^{\mathscr{B}(\bz_n)}_i$.\footnote{It is straightforward to show that the type sets $T_i^{\mathscr{B}(\bz_n)}$ are \emph{belief-closed}; that is,  $\kappa^{\mathscr{B}(\bz_n)}_i(t_i)\left(\Theta \times T_{-i}^{\mathscr{B}(\bz_n)}\right)=1$ for every $t_i \in T_i^{\mathscr{B}(\bz_n)}$.
} The type $t_i$ is \emph{permitted} for player $i$ if $t_i \in T^{\mathscr{B}(\bz_n)}_i$.
\end{definition}

This type space includes all type profiles where each player $i$ has common certainty in the event that all players have first-order beliefs in $\mathscr{B}(\bz_n)$. Note that the type space permits common knowledge disagreement\textemdash  that is, player $i$ can believe with probability 1 that (all believe with probability 1 that...) players hold different first-order beliefs. Such types are precluded under the common prior assumption not only in the present setting of common data, but also if we were to allow for private and different information \citep{disagree}. The induced type spaces in Definition \ref{def:typespace} thus corresponds to a relaxation of the common prior assumption, where the permitted extent of disagreement is governed by the set of learning rules $\mathcal{M}$. In the special case in which $\mathcal{M}$ consists of a singleton Bayesian rule, then we return the common prior assumption.

In this approach, restrictions are placed \emph{only} on the beliefs that players hold on the exogenous parameter space $\Theta$, and not on how they came to form those beliefs.  For example, each of the following is consistent with the restriction in Definition \ref{def:typespace}:
\begin{itemize}
\item (Common Knowledge) Each player $i$ is associated with a player-specific learning rule $\mu_i \in \mathcal{M}$. It is common knowledge that given any data set $\bz_n$, each player $i$'s first-order belief is $\mu_i(\bz_n)$. 
\item (Randomization) Each player $i$ randomizes over learning rules in $\mathcal{M}$ using a player-specific distribution $Q_i \in \Delta(\mathcal{M})$, and applies the randomly drawn learning rule to the realized data to form a first-order belief. The distributions $(Q_i)_{i \in \mathcal{M}}$ are common knowledge, although the realized random rule is privately known.
\item (Misspecification.) Player $j$ believes with probability 1 that player $i$'s first order belief is $\mu(\bz_n)$ for every data set $\bz_n$, whereas in fact player $i$'s first-order belief is $\mu'(\bz_n)$ for a different learning rule $\mu'\in \mathcal{M}$. 
\end{itemize}

It is possible sometimes to take $\mathscr{B}(\bz_n)$ as the primitive without explicitly defining $\mathcal{M}$ (as in  the coordination example in Section \ref{sec:ExampleCoord}), so that each measure $P^n$ directly defines a measure over sets of first-order beliefs.\footnote{The definition of $\mathscr{B}(\bz_n)$ must, however, respect Assumption \ref{ass:limitbelief}. However, there must exist some set of learning rules $\mathcal{M}$ that could be defined, which respects Assumption \ref{ass:limitbelief} and gives rise to these sets $\mathscr{B}(\bz_n)$. A sufficient condition is for there to exist a $P$-measure 1 set of sequences $\mathcal{Z}^* \subseteq \mathcal{Z}^\infty$ such that for every $\bz \in \mathcal{Z}^*$, and every sequence $(\nu^n)_{n=1}^\infty$ satisfying $\nu^n \in \mathscr{B}(\bz_n)$ for each $n$, it holds that $\lim_{n\rightarrow \infty} d_P(\nu^n, \mu^\infty)=0$.} In other cases, it is more natural to define $\mathcal{M}$. Some examples for this set include: 

\bigskip

\noindent \emph{Bayesian Updating with Different Priors.}  Each learning rule $\mu_\pi \in \mathcal{M}$ is identified with a prior distribution $\pi \in \Delta(\Theta \times \mathcal{Z}^\infty)$. For any realized data $\bz_n$, the belief $\mu_\pi(\bz_n)$ is the marginal over $\Theta$ of the posterior belief associated with the corresponding prior $\pi$.

\bigskip

\noindent \emph{Sample Statistics.} The set $\mathcal{M}$ consists of learning rules that map the data to different point-estimates for the payoff-relevant parameter. For example, $\mathcal{M}$ might consist of the two learning rules $\mu_{\text{mean}}$ and $\mu_{\text{median}}$, where for any data set $\bz_n$, $\mu_{\text{mean}}(\bz_n)$ is a point-mass belief on the mean realization in $\bz_n$ (as in \citet{Jehiel2018}), and $\mu_{\text{median}}(\bz_n)$ is a point-mass belief on the median realization. 
\bigskip

\noindent \emph{Linear Regression.} Suppose that  $X \subseteq \mathbb{R}^p$, $p<\infty$, is a set of attributes that determine the value of a parameter in $\Theta$ (e.g. physical covariates of a patient seeking health insurance, and medical outcomes for those patients). The observations in $\bz_n$ are pairs $(x,\theta)$, and the payoff-relevant unknown is the parameter associated with some new $x^\ast$ (e.g. the outcome for a new patient with characteristics $x^\ast$).

Each learning rule $\mu\in \mathcal{M}$ corresponds to a different regression model based on a subset of attributes $I_\mu \subseteq \{1, \dots, p\}$ (as for example in \citet{competingmodels}). Write $x_\mu = (x_i)_{i \in I_\mu}$ for the coordinates of $x$ at those indices. Then, $\hat{f}^{OLS}_\mu[\bz_n](x) = \beta^{OLS}_\mu \cdot x_\mu$ is the linear function of the attributes in $I_\mu$ that best fits the observed data $\bold{z}_n = \{(x^k, \theta^k)\}_{k=1}^n$; that is,
$\beta^{OLS}_\mu = \argmin_{\beta \in \mathbb{R}^{\vert I_\mu \vert}} \frac1n \sum_{i=1}^n (\beta \cdot x^k_\mu - \theta^k)^2$. 
For every $\bz_n$, the learning rule $\mu$ maps $\bz_n$ into a point-mass belief on the corresponding prediction $\hat{f}^{OLS}_\mu[\bz_n](x^\ast)$. 

\bigskip

\noindent \emph{Case-Based Learning with Different Similarity Functions.} As in the previous example, suppose that the observations are pairs $(x,\theta)\in X \times \Theta$, and the payoff-relevant parameter is the outcome at some new $x^\ast$. Each ``case-based" learning rule $\mu \in \mathcal{M}$ is identified with a real number $\lambda \in \mathbb{R}_+$ (to be interpreted momentarily) and maps the historical data $\bold{z}_n = \{(x^k, \theta^k)\}_{k=1}^n$ into a weighted average of the observed parameter values $\theta^k$ \citep{casebased}.

The observed parameters at $x$-values ``more similar" to $x^\ast$ are weighted more heavily. Formally, let $g: X \times X \rightarrow \mathbb{R}_+$ be a similarity function on attributes, where $g(x,x')$ describes the distance between attribute vectors $x$ and $x'$. The learning rule with parameter $\lambda$ maps $\bold{z}_n $ to a point mass on the weighted average $\frac1n \sum_{k=1}^n \theta^k \cdot  (e^{-\lambda g_\mu(x^k, x^\ast)})/(\sum_{k'} e^{-\lambda g_\mu(x^{k'},x^\ast)}).$
The parameter $\lambda$ controls the degree to which similar observations are weighted more heavily than dissimilar observations. For example, $\lambda = 0$ returns a simple average of all of the observed states, while $\lambda \rightarrow \infty$ returns the observed state at the most similar attribute vector. 

\subsection{Analyst's Confidence Set} \label{sec:robustplausible}
I now use the proposed framework to construct a quantitative metric for the analyst's confidence in a strategic prediction, focusing on prediction that an action is \emph{interim-correlated rationalizable} \citep{ICR,WeinsteinYildiz2017}.
The use of this particular solution concept is not critical to the approach---for example, we could study the agent's confidence in prediction of Bayesian Nash equilibria (see Section \ref{sec:extensions})---but interim-correlated rationalizability is well-suited to the present setting, where agents may have common knowledge disagreement.\footnote{Equilibrium notions are known to lead to potentially counterintuitive predictions when players have common knowledge disagreement. For example, consider a matching pennies game where player 1 receives $\theta$ if players match and $-\theta$ otherwise, and player 2 receives $-\theta$ if the players match and $\theta$ otherwise. Let $\theta \in \{-1,1\}$. Then if player 1 assigns probability 1 to $\theta=1$ while player 2 assigns probability 1 to $\theta=-1$, it is (somewhat counterintuitively) a Bayesian Nash equilibrium for both players to choose match. See \citet{learningbayesian} for an extended discussion.}  Its definition is reviewed here:

For every player $i$ and type $t_i \in T^*_i$, set $S_i^0[t_i]=A_i$, and define $S_i^k[t_i]$ for $k\geq 1$ such that $a_i \in S_i^k[t_i]$ if and only if $a_i$ is a best reply to some $\pi \in \Delta(\Theta \times T^*_{-i} \times A_{-i})$ satisfying (1) $\marg_{\Theta \times T^*_{-i}} \pi = \kappa^*_i(t_i)$ and (2) $\marg_{A_{-i} \times T_{-i}} \pi(\{a_{-i},t_{-i}) \mid a_{-i} \in S^{k-1}_{-i}[t_{-i}]\})=1$, where $S_{-i}^{k-1}[t_{-i}]= \prod_{j \neq i} S_j^{k-1}[t_{-j}]$. We can interpret $\pi$ to be an extension of type $t_i$'s belief $\kappa^*_i(t_i)$ onto the space $\Delta(\Theta \times T_{-i} \times A_{-i})$, with support in the set of actions that survive $k-1$ rounds of iterated elimination of strictly dominated strategies for types in $T_{-i}$. For every $i$, the actions in  
$S_i^\infty [t_i] = \bigcap_{k=0}^\infty S_i^k[t_i]$
are \emph{interim correlated rationalizable} for type $t_i$, or (henceforth) simply \emph{rationalizable}. Say that $a_i$ is \emph{strictly} rationalizable for type $t_i$ if the best reply conditions above are strengthened to strict best replies.

For any set of beliefs $\mathscr{B} \subseteq \Delta(\Theta)$, say that action $a_i$ is \emph{strongly $\mathscr{B}$-rationalizable} if it is rationalizable for player $i$ with \emph{any type} $t_i \in T_i^{\mathscr{B}}$, and it is \emph{weakly $\mathscr{B}$-rationalizable} if it is rationalizable for player $i$ with \emph{some type} $t_i \in T_i^{\mathscr{B}}$.\footnote{Whether an action $a_i$ is interim-correlated rationalizable for some type $t_i$ does not depend on the description of the underlying type space \citep{ICR}. Hence we only need to define rationalizability for the universal type space, even though the type spaces that we will work with are the smaller type spaces $\left(T_i^{\mathscr{B}(\bz_n)}, \kappa_i^{\mathscr{B}(\bz_n)}\right)$, and they vary depending on the data $\bz_n$.} Strong and weak $\mathscr{B}$-rationalizability represent (respectively) maximally stringent and maximally lenient approaches for determining whether $a_i$ constitutes a ``reasonable" prediction in the interim type space $\left(T^{\mathscr{B}}_i, \kappa^{\mathscr{B}}_i\right)_{i \in \mathcal{I}}$.\footnote{Given the restriction in Definition \ref{def:typespace}, the weakly $\mathscr{B}$-rationalizable strategies are the $\Delta$-rationalizable strategies of \citet{battigalli}, where $\Delta=(\Delta_i)_{i \in \mathcal{I}}$ and each $\Delta_i = \{\nu \in \Delta(\Theta \times T_{-i} \times A_{-i}) \mid \marg_{\Theta} \mu \in \mathscr{B}\}$ encodes the belief restriction that first-order beliefs belong to $\mathscr{B}$. The concept of strong $\mathscr{B}$-rationalizability can be interpreted as a ``robust" version of $\Delta$-rationalizability.}

The main concept of a confidence set is now defined.

\begin{definition} For every $n \in \mathbb{Z}_+$, define $\underline{p}^n(a_i)$ to be the probability (over possible datasets $\bz_n$) that action $a_i$ is rationalizable \emph{for every} type in $T_i^{\mathscr{B}(\bz_n)}$; that is,
\begin{equation} \label{pRobust}
\underline{p}^n(a_i)= P^n \left(\{\bz_n\,: \, a_i  \mbox{ is strongly $\mathscr{B}(\bz_n)$-rationalizable}\}\right).
\end{equation}
Define $\overline{p}^n(a_i)$ to be the probability (over possible datasets $\bz_n$) that action $a_i$ is rationalizable \emph{for some} type $t_i \in T^{\mathscr{B}(\bz_n)}_i$; that is,
\begin{equation} \label{pPlausible}
\overline{p}^n(a_i)= P^n \left(\{\bz_n\,: \, a_i \mbox{ is weakly $\mathscr{B}(\bz_n)$-rationalizable}\}\right).
\end{equation}
The \emph{confidence set} for rationalizability of $a_i$ given $n$ observations is $[\underline{p}^n(a_i),\overline{p}^n(a_i)]$.
\end{definition}

The larger $\underline{p}^n(a_i)$ and $\overline{p}^n(a_i)$ are, the more confident an analyst should be in predicting that $a_i$ is rationalizable. At extremes: If $\overline{p}^n(a_i)=\underline{p}^n(a_i)=1$, then given observation of $n$ random samples,  action $a_i$ is guaranteed to be rationalizable for player $i$ (for all permitted types). If $\overline{p}^n(a_i)=\underline{p}^n(a_i)=0$, then action $a_i$ is guaranteed to not be rationalizable for player $i$ (for any permitted types). In the intermediate cases, if $0<\underline{p}^n(a_i)=\overline{p}^n(a_i)<1$, then rationalizability of the action $a_i$ depends on the specific realization of the data, and if $\underline{p}^n(a_i) <\overline{p}^n(a_i)$, then the prediction requires assumptions on the details of the agent's belief beyond Assumption \ref{ass:Example}.\footnote{I do not comment here on what further assumptions may be imposed, interpreting this case simply as one in which the prediction is tenuous.}

\begin{observation} For every player $i$ and action $a_i\in A_i$:
\begin{itemize}
\item[(a)] $\underline{p}^n(a_i) \leq \overline{p}^n(a_i)$ for every $n \in \mathbb{Z}_+$.
\item[(c)] If $\mathcal{M}$ consists of a single learning rule,  then $\underline{p}^n(a_i) = \overline{p}^n(a_i)$ for every $n \in \mathbb{Z}_+$. 
\end{itemize}
\end{observation}

In the special case in which agents have a common prior, the definitions in $\underline{p}^n(a_i)$ and $\overline{p}^n(a_i)$ have the following familiar interpretation:

\begin{remark}(Common Prior.) \label{ex:CP} Suppose that players share a common prior over $\Theta \times \mathcal{Z}^\infty$ and for simplicity let $\mathcal{Z}$ be finite. Write $\mu$ for the learning rule that maps $\bz_n$ into the induced posterior belief over $\Theta$ under the common prior. Then, each realization $\bz_n$ determines an interim game, where players all have common certainty in the posterior belief. Moreover, the common prior determines a distribution over $\bz_n$, and hence over possible interim games. For any player $i$ and action $a_i$, the probabilities $\underline{p}^n(a_i) = \overline{p}^n(a_i)$, and are equal to the measure of size-$n$ datasets $\bz_n$ (under the common prior\footnote{A small difference in the formulations is that  $\underline{p}^n(a_i)$ and $\overline{p}^n(a_i)$ are defined using the ``true" probability measure $P \in \Delta(\mathcal{Z}^\infty)$ in the present approach, instead of a measure $Q \in \Delta(\Theta \times \mathcal{Z}^\infty)$.}) with the property that action $a_i$ is rationalizable for player $i$ in the corresponding interim game.\footnote{This approach is used for example in \citet{kajiimorris} (if we re-interpret the histories $\bz_n$ as the states), where an incomplete information game is ``close" to a complete information game if the payoffs of the complete information game occur with high probability under the prior.}

\end{remark}

In the above approach, the common prior serves multiple roles: it determines the true distribution over the data that agents might see, and also determines how agents update from that data. When we separate these roles, we can still use an objective data-generating process to define a measure over interim games, as I do here. In this way, the probabilities $\underline{p}^n(a_i)$ and $\overline{p}^n(a_i)$ are a natural generalization of a standard measure of the typicality of a strategic prediction, in the absence of a common prior.

\section{Asymptotic Results} \label{sec:Asymptotic}

The subsequent sections study how confidence sets depend on the underlying learning environment and the game in question. I first consider the limiting behavior of the probabilities $\overline{p}^n(a_i)$ and $\underline{p}^n(a_i)$ as the quantity of data $n$ gets large. Recall that by Assumption \ref{ass:limitbelief}, the beliefs induced by the different learning rules converge to a limiting belief $\mu^\infty$. Thus, the $n=\infty$ limit corresponds to an incomplete information game in which players have common certainty in the event that every player has first-order belief $\mu^\infty$. Whether the probabilities $\underline{p}^n(a_i)$ and $\overline{p}^n(a_i)$ are continuous at $n=\infty$ tells us how sensitive rationalizability of $a_i$ is to an assumption that agents have coordinated their beliefs using infinite data. When these probabilities are discontinuous at $n=\infty$, then the infinite-data prediction is fragile---that is, the analyst would make a different prediction for arbitrarily large but finite quantities of data.

Formally, let $t_i^\infty$ be the player $i$ type with common certainty in the event that each player's first-order belief is $\mu^\infty$. Then define $\underline{p}^\infty(a_i)=\overline{p}^\infty(a_i)=1$ if the action $a_i$ is rationalizable for type $t_i^\infty$, and define $\underline{p}^\infty(a_i)=\overline{p}^\infty(a_i)=0$ if it is not. 

\begin{definition} The confidence set for action $a_i$ is \emph{asymptotically continuous} if 
\[\lim_{n \rightarrow \infty} [\underline{p}^n(a_i),\overline{p}^n(a_i)] = [\underline{p}^\infty(a_i), \overline{p}^\infty(a_i)].\]
\end{definition}

\subsection{Fragile Predictions}

Whether confidence sets are asymptotically continuous depends crucially on whether the beliefs induced by  the different learning rules converge \emph{uniformly} to $\mu^\infty$. 
\begin{assumption}[Uniform Convergence] \label{ass:UniformConv} 
$\lim_{n\rightarrow \infty} \sup_{\mu \in \mathcal{M}} d_P(\mu(Z^n),\mu^\infty)=0$ $P$-a.s., where $d_P$ is the Prokhorov metric on $\Delta(\Theta)$.
\end{assumption}

Assumption \ref{ass:limitbelief} already implies that for each learning rule $\mu \in \mathcal{M}$, the (random) induced belief $\mu(Z^n)$  almost surely converges to $\mu^\infty$ as the quantity of data $n$ grows large. Assumption \ref{ass:UniformConv} strengthens this by requiring additionally that the speed of convergence does not vary too much across the different learning rules in $\mathcal{M}$. Specifically, the sequence of beliefs $\{\mu(Z^n)\}$ must converge to $\mu^\infty$ (as $n \rightarrow \infty$) uniformly across $\mu \in \mathcal{M}$. 

A sufficient condition for Assumption \ref{ass:UniformConv} to hold is that the set of learning rules $\mathcal{M}$ is finite. But failures of Assumption \ref{ass:UniformConv} occur for classes of learning rules that we may consider plausible. In particular, Assumption \ref{ass:UniformConv} fails if the class $\mathcal{M}$ is too rich, as in the following example:

\begin{example}[Rich Sets of Priors and Likelihoods] \label{ex:RichPriors} An unknown parameter $v$ takes values in $\{0,1\}$. Players commonly observe a sequence of realizations from the set $Z=\{0,1\}$. Learning rules $\mu_{\pi,q} \in \mathcal{M}$ are indexed to parameters $\pi \in (0,1)$ and $q\in (1/2,1)$, where the parameter $\pi$ is the prior probability of value 1, and $q$ identifies the following signal structure:
\[\begin{array}{ccc}
& z=0 & z=1 \\
v=0 & q & 1-q \\
v=1 & 1-q & 1
\end{array}\]
 Each rule $\mu_{\pi,q}$  is identified with prior $\pi$ and signal structure $q$, and maps the observed signal outcomes into the posterior belief over $\{0,1\}$. Assume that the true data-generating process belongs to this class; that is, there exists some  $q^* \in (1/2,1)$ such that  the distribution over the signal set $\{0,1\}$ is $(q^*,1-q^*)$ when $v=0$, and the distribution is $(1-q^*,q^*)$ when $v=1$.
\end{example}

In this example, all learning rules lead to the same belief (that is, there is \emph{asymptotic agreement} in the sense of \citet{yildiz}). But because the rate of this convergence cannot be uniformly bounded across the different learning rules, it is possible for the confidence set to be discontinuous at $n=\infty$.

\begin{claim} \label{claim:Discontinuity} 
Consider the trading game described in Section \ref{sec:example}, and the data-generating process and set of learning rules from Example \ref{ex:RichPriors}. Then,
$\lim_{n\rightarrow \infty} [\underline{p}^n(a_i),\overline{p}^n(a_i)]=[0,1]$, while 
 $[\underline{p}^\infty(a_i),\overline{p}^\infty(a_i)]=\{0\}$, so the prediction that entering is not rationalizable for the Seller is not asymptotically continuous.
\end{claim}

The claim tells us that although trade will not occur in the limiting game, this prediction is  sensitive to the assumption that agents have indeed coordinated their priors using infinite data. Even if the amount of data that players commonly observed were to be arbitrarily large, the analyst should nevertheless consider trade to be a plausible outcome.

\subsection{Asymptotic Continuity} \label{sec:Continuity}

In contrast, when the assumption of uniform convergence is satisfied, then the limiting confidence sets can be tightly linked to predictions in the limiting game.

\begin{theorem} \label{prop:Limit} Suppose Assumption \ref{ass:UniformConv} is satisfied. 
\begin{itemize}
\item[(a)] If $a_i$ is strictly rationalizable for player $i$ of type $t_i^\infty$, then
\[lim_{n\rightarrow \infty}[\underline{p}^n(a_i),\overline{p}^n(a_i)]=\{1\}.\]
\item[(b)] If $a_i$ is not rationalizable for player $i$ of type $t_i^\infty$, then
\[\lim_{n\rightarrow \infty}[\underline{p}^n(a_i),\overline{p}^n(a_i)] = \{0\}.\]

\end{itemize}
\end{theorem}

This theorem says that if an action $a_i$ is strictly rationalizable for player $i$ given infinite data, then $\overline{p}^n(a_i)$ and $\underline{p}^n(a_i)$ both converge to 1 as $n$ grows large.\footnote{If the limiting belief $\mu^\infty$ is degenerate at a limiting parameter $\theta^\infty$, and players have common certainty that players' first-order beliefs have support in a shrinking neighborhood of $\theta^\infty$ (see Section \ref{subsec:lower} for a more formal development), then the property that $\overline{p}^n(a_i)\rightarrow 1$ is  equivalent to the property that $a_i$ is robustly rationalizable, as defined in \citet{satoru1}, with the small difference that \citet{satoru1} consider almost common belief in the exact parameter $\theta^\infty$, while I consider common certainty in a neighborhood of $\theta^\infty$. As Proposition 1 in \citet{satoru1} shows, strict rationalizability is a sufficient condition for robust rationalizability. See also \citet{KajiiMorris2020} for related results.} Thus, when agents observe sufficiently large quantities of public data, the analyst should be arbitrarily confident in predicting that $a_i$ is rationalizable. On the other hand, if action $a_i$ is not rationalizable given infinite data, then $\overline{p}^n(a_i)$ and $\underline{p}^n(a_i)$ both converge to 0, so the analyst should be arbitrarily confident in predicting that $a_i$ is \emph{not} rationalizable for large data sets.\footnote{The intermediate case in which $a_i$ is rationalizable for player $i$ given infinite data, but not strictly rationalizable,  is subtle and depends on details of the game. See Online Appendix \ref{app:More} for examples in which $\lim_{n\rightarrow \infty} [\underline{p}^n(a_i),\overline{p}^n(a_i)] = \{1\}$ and in which $\lim_{n\rightarrow \infty}[\underline{p}^n(a_i),\overline{p}^n(a_i)]=[0,1]$. Note that the latter corresponds to a maximally ambiguous outcome\textemdash no amount of data is decisive on whether or not the action should be considered rationalizable.}

Theorem \ref{prop:Limit} builds on results from the literature on topologies on the universal type space. Consider any sequence of types $(t_i^n)_{n=1}^\infty$ where each $t_i \in T_i^{\mathscr{B}(\bz_n)}$. Under Assumption \ref{ass:UniformConv}, as the quantity of data $n$ gets large, the types $t_i^n$ (almost surely) have common certainty that first-order beliefs lie in an arbitrarily small  neighborhood of the limiting belief $\mu^\infty$. Thus, the sequence $(t_i^n)$ can be shown to converge  to $t_i^\infty$, in the \emph{uniform-weak topology} \citep{faingold} on the universal type space (see Lemma \ref{lemm:UW}). Since rationalizability is upper hemi-continuous in the uniform-weak topology \citep{faingold}, Part (b) of the theorem follows.

Part (a) of the theorem is related to lower hemi-continuity of strict rationalizability in the uniform-weak topology (as shown in \citet{faingold}),\footnote{It is crucial that convergence occurs in this topology and not simply the product topology, as otherwise the negative results of \citet{WY} would apply.} but this property is not sufficient. Lower hemi-continuity guarantees that for any sequence of types $(t_i^n)_{n=1}^\infty$ from $T_i^{\mathscr{B}(\bz_n)}$, the action $a_i$ must eventually be rationalizable along the sequence, but the rates of this convergence can differ substantially across different sequences. For eventual strong $\mathscr{B}(\bz_n)$-rationalizability, we need that $a_i$ is rationalizable for \emph{all} types from $T_i^{\mathscr{B}(\bz_n)}$ when $n$ is sufficiently large. To establish this, I show that there is a $P$-measure 1 set of sequences along which the sets $\left(T_i^{\mathscr{B}(\bz_n)}\right)_{n=1}^\infty$ converge to the singleton set $\{t_i^\infty\}$ in the Hausdorff metric induced by the uniform-weak metric. The key lemma underlying this result, Lemma \ref{lemm:eps}, relates the degree of ``strictness" of rationalizability of action $a_i$ at the limiting type $t_i^\infty$ to the size of the neighborhood around $\mu^\infty$ such that common certainty of that neighborhood implies rationalizability of $a_i$. The stronger property that types converge uniformly over the set $T_i^{\mathscr{B}(\bz_n)}$ delivers the desired result.

\section{(Small) Finite Samples} \label{sec:finite}

The previous section characterized confidence sets given large numbers of common observations. I now focus on the setting of small $n$, and bound the extent to which the agent's confidence set $[\underline{p}^n(a_i),\overline{p}^n(a_i)]$ diverges from its asymptotic limit $[\underline{p}^\infty(a_i),\overline{p}^\infty(a_i)]$. Throughout this section, I impose the simplifying assumptions that observations are i.i.d., and that they take values from a finite set $\mathcal{Z}$:

\begin{assumption} \label{iid} $Z_1, \dots, Z_n \sim_{\text{i.i.d.}} Q$.
\end{assumption}
\begin{assumption}  $\vert \mathcal{Z} \vert <\infty$.
\end{assumption}

In some cases, as in the examples in Section \ref{sec:example}, the confidence set can be exactly characterized. In what follows, I provide bounds for the confidence set that can be easier to derive in certain cases.

\subsection{Lower Bound} \label{subsec:lower}
First consider an action $a_i$ that is strictly rationalizable for player $i$ of type $t_i^\infty$. By Theorem \ref{prop:Limit}, the analyst's confidence set $[\underline{p}^n(a_i),\overline{p}^n(a_i)]$ converges to a degenerate interval at 1. Proposition \ref{prop:speedRobust}, below, provides a lower bound on $\underline{p}^n(a_i)$, which informs how fast this convergence occurs.

A key input into the bound is the ``degree" to which $a_i$ is strictly rationalizable for the limiting type $t_i^\infty$. Say that a family of sets $(R_i[t_i])_{t_i \in T_i}$, where each $R_j[t_j]\subseteq A_j$, has the \emph{$\delta$-strict best reply property} if for each $i\in \mathcal{I}$, type $t_i \in T_i$, and action $a_i \in R_i[t_i]$ there is a conjecture $\sigma_{-i}: \Theta \times T_{-i} \rightarrow \Delta(A_{-i})$ to which $a_i$ is a $\delta$-strict best reply for $t_i$; that is,
\[\int_{\Theta} u_i(a_i, \sigma_{-i}(\theta,t_{-i}), \theta) t_i[d\theta \times dt_{-i}] -\int_{\Theta} u_i(a'_i, \sigma_{-i}(\theta,t_{-i}), \theta) t_i[d\theta \times dt_{-i}] \geq \delta \quad \forall a_i'\neq a_i.\]
Say that an action $a_i$ is \emph{$\delta$-strict rationalizable} for type $t_i$ if there exists a family of sets  $(R_j[t_j])_{t_j \in T_j}$ with the $\delta$-strict best reply property, where $a_i \in R_i[t_i]$.\footnote{This is equivalent to $\gamma$-rationalizability from \citet{ICR}, where $\gamma=-\delta$.}

Then, if $a_i$ is strictly rationalizable for the limiting type $t_i^\infty$, and players have commonly observed $n$ realizations, the probability that $a_i$ is rationalizable for all permitted types can be upper bounded as follows. 

\linespread{1}
\begin{proposition} \label{prop:speedRobust}  Suppose $a_i$ is strictly rationalizable for type $t_i^\infty$, and define
  \begin{equation} \label{eq:deltainf}
  \delta^\infty := \sup \left\{ \delta \,: \, a_i \mbox{ is $\delta$-strictly rationalizable for type $t_i^\infty$}\right\}
  \end{equation}
  noting that this quantity is strictly positive. Further define \begin{equation} \label{eq:xi}
\xi := \sup_{\theta, \theta' \in \Theta} \| \theta - \theta' \|_\infty.
\end{equation}
Then, for every $n\geq 1$,
\begin{equation} \label{boundR}
\underline{p}^n(a_i) \geq 1-\frac{2K\xi}{\delta^\infty} \mathbb{E}\left(\sup_{\mu \in \mathcal{M}} d_P(\mu(Z^n),\mu^\infty)\right)\end{equation}
where $K$ is the Lipschitz constant of the map $g: \Theta \rightarrow U$.
\end{proposition}
\linespread{1.35}
\noindent Recalling that $\overline{p}^n(a_i)\geq \underline{p}^n(a_i)$ for every $n$, this proposition allows us to lower bound the confidence set $[\underline{p}^n(a_i),\overline{p}^n(a_i)]$.

The expression in (\ref{boundR}) is increasing in $\delta^\infty$, so the ``more strictly-rationalizable" the action is for the limiting type, the fewer observations are necessary for the prediction to hold. The bound is decreasing in $\mathbb{E}(\sup_{\mu \in \mathcal{M}} d_P(\mu(Z^n),\mu^\infty))$, which is the expected distance from the limiting belief $\mu^\infty$ to the farthest belief in the plausible set $\mathscr{B}(Z^n)$. When Assumption \ref{ass:UniformConv} is satisfied, then $\mathbb{E}(\sup_{\mu \in \mathcal{M}} d_P(\mu(Z^n),\mu^\infty)) \rightarrow 0$ as $n \rightarrow \infty$, and the speed of this convergence can be interpreted as the speed at which players \emph{commonly learn} \citep{commonlearning}. Thus, Theorem \ref{prop:speedRobust} suggests that the quicker players commonly learn, the fewer observations are necessary for limiting predictions to carry over to small-data settings.

In an important special case, the limiting belief $\mu^\infty$ is a point mass at some $\theta^\infty$,  and the sets $\mathscr{B}(\bz_n)$ consist of beliefs with support on shrinking neighborhoods of $\theta^\infty$. Formally, let 
\[\mathscr{C}(\bz_n) := \bigcup_{\mu \in \mathcal{M}} \supp \mu(\bz_n)  \quad  \forall \, \bz_n \in \mathcal{Z}^n\]
with the implication that every $\mu(\bz_n)$, $\mu \in \mathcal{M}$,  assigns probability 1 to $\mathscr{C}(\bz_n)$. If $\mathscr{C}(\bz_n)$ collapses to the singleton set $\{\theta^\infty\}$ as $n\rightarrow \infty$, then the bound in Proposition \ref{prop:speedRobust} can be simplified as follows.

\begin{assumption} \label{ass:ShrinkSupport}
$ \sup_{\theta \in \mathscr{C}(Z^n)}\|\theta - \theta^\infty \|_\infty$ converges to zero $P$-almost surely.\footnote{That is, there is a $P$-measure 1 set of (infinite) sequences such that $ \sup_{\theta \in \mathscr{C}(\bz_n)}\|\theta - \theta^\infty \|_\infty\rightarrow 0$ as $n \rightarrow \infty$ for each sequence $\bz$ in this set.}
\end{assumption}

\begin{proposition}  \label{corr:ShrinkSet} Suppose Assumption \ref{ass:ShrinkSupport} holds, and the action $a_i$ is strictly rationalizable for type $t_i^\infty$. Then, for every $n\geq 1$,
\linespread{1}
\begin{equation} \label{boundR}
\underline{p}^n(a_i) \geq 1- \frac{2K}{\delta^\infty}  \mathbb{E} \left(  \sup_{\theta\in \mathscr{C}(Z^n)} \left\| \theta - \theta^\infty \right\|_\infty\right)\end{equation}
where $K$ is the Lipschitz constant of the map $g: \Theta \rightarrow U$.
\end{proposition}
\linespread{1.35}

The expressions in Propositions \ref{prop:speedRobust} and \ref{corr:ShrinkSet} can be used to derive quantitative bounds for specific sets of learning rules, as in the following example: 

\begin{example} Consider the payoff matrix from Section \ref{sec:ExampleCoord} with unknown parameter $\beta \in \mathbb{R}$. Suppose that players commonly observe $n$ public signals
$z_t=\beta+\varepsilon_t$,
with standard normal error terms $\varepsilon_t$ that are i.i.d. across observations. The set of learning rules is $\mathcal{M}=\{\mu_x\}_{x \in [-\eta,\eta]}$, where each learning rule $\mu_x$ is identified with the prior belief
$\beta \sim \mathcal{N}(x,1)$, and maps data into a point mass at the posterior expectation of $\beta$. The set $\mathscr{C}(\bz_n)$ thus consists of the posterior expectations under the different priors, and players have common certainty in the event that all players have first-order beliefs with support on $\mathscr{C}(\bz_n)$. Let the true value of $\beta$ satisfy $\beta>1$. Then, applying Proposition \ref{corr:ShrinkSet}:

\begin{corollary} \label{corr:normal} For each $n\geq1$,
\linespread{1}
\[\underline{p}^n(\emph{strong}) \geq 1-\frac{1}{\beta-1} \left(\sqrt{\frac{2}{\pi n}} + \frac{\beta+\eta}{n+1}\right)\]
\end{corollary}
\linespread{1.35}
\noindent The bound in Corollary \ref{corr:normal} is decreasing in $\eta$ (the size of the model class), increasing in $n$ (the number of observations), and increasing in $\beta-1$ (the strictness of the solution at the limit). 

\end{example}

\subsection{Upper Bound}
Now suppose that the action $a_i$ is \emph{not} rationalizable for player $i$ of type $t_i^\infty$. We know from Part (c) of Theorem \ref{prop:Limit} that in this case, the analyst's confidence set $[\underline{p}^n(a_i),\overline{p}^n(a_i)]$ converges to a degenerate interval at zero. But given small quantities of data $n$, the action $a_i$ may still constitute a plausible prediction of play, as in the trading game studied in Section \ref{sec:example}. Claim \ref{prop:Sanov}, below, provides an upper bound on $\overline{p}^n(a_i)$, which informs whether the analyst should consider $a_i$ a plausible prediction away from the limit.

To define this bound, a few intermediate definitions are needed. Let $\mathbb{Z}_{a_i}$
 be all data sets $\bz_n$ given which the action $a_i$ is weakly $\mathscr{B}(\bz_n)$-rationalizable. (This set must be determined on a case-by-case basis.) Let $\widehat{Q}_{\bz_n} \in \Delta(\mathcal{Z})$ be the empirical measure associated with data set $\bz_n$. The Kullback-Leibler divergence between $\widehat{Q}_{\bz_n}$ and the actual data-generating distribution $Q$ is
$D_{KL}(\widehat{Q}_{\bz_n} \| Q) = \sum_{z \in \mathcal{Z}} Q(z) \log \left(\frac{\widehat{Q}_{\bz_n}(z)}{Q(z)}\right)$. Define
\[Q^*_n = \argmin_{\widehat{Q}\in \left\{\bz_n \in \mathbb{Z}^n_{a_i}\right\}} D_{KL}(\widehat{Q}_{\bz_n} \| Q)\]
to be the empirical measure (associated with a data set in $\mathbb{Z}_{a_i}$) that is closest in Kullback-Leibler divergence to $Q$.  Application of Sanov's theorem directly gives the following result.

\begin{proposition} \label{prop:Sanov} Suppose $a_i$ is not rationalizable for type $t_i^\infty$; then, for every $n\geq 1$,
\[\overline{p}_n(a_i) \leq (n+1)^{|\mathcal{Z}|} 2^{-n D_{KL}\left(Q_n^*\| Q\right)}.\]
\end{proposition}

\noindent Recalling that $\overline{p}^n(a_i)\geq \underline{p}^n(a_i)$ for every $n$, this proposition allows us to upper bound the confidence set $[\underline{p}^n(a_i),\overline{p}^n(a_i)]$. The claim is applied below in an example setting:

\begin{example}
Consider the trading game from Section \ref{sec:example} and the learning rules described in Example \ref{ex:RichPriors}, but suppose that the domain of $q$ is $[2/3,1]$ and the domain of $\pi$ is $[1/4,3/4]$, so that Assumption \ref{ass:UniformConv} is satisfied. Let the true signal structure be identified with $q^*=3/4$. and suppose the posted price is $p=3/4$. Theorem \ref{prop:Limit} implies that entering will fail to be rationalizable when players have observed sufficient data. Nevertheless, the action may be rationalizable for a permitted belief if players have observed a small number of data points. The corollary below quantifies this probability. 

\begin{corollary}  \label{corr:Sanov} For each $n\geq1$,
\[\overline{p}^n(\mbox{enter}) \leq (n+1)^2 2^{-r_n n}\]
where $r_n = \frac34 \left(\log(3n) - \log\left(\lfloor \frac{n}{2} + \frac{\log(9)}{\log(2)}\rfloor \right)\right) + \frac14 \left(\log(n) - \log\left(\lfloor \frac{n}{2} - \frac{\log(9)}{\log(2)}\rfloor \right)\right) $.
\end{corollary}

\end{example}

\section{Extensions} \label{sec:extensions}

\paragraph{Asymptotic Disagreement.} In the main text, I imposed an assumption which guaranteed that beliefs produced by learning rules in $\mathcal{M}$ uniformly converge to a common limiting belief $\mu^\infty$. This implies that learning eventually removes all differences in beliefs. It is possible to replace Assumption \ref{ass:UniformConv} with the following, weaker condition, which allows players to have heterogeneous beliefs even in the limit: For any $\epsilon \geq 0$, say that the class of learning rules $\mathcal{M}$ satisfies $\epsilon$-Uniform Convergence if 
\[\lim_{n\rightarrow \infty} \sup_{\mu \in \mathcal{M}} d_P(\mu(Z^n),\mu^\infty) \leq \epsilon \quad \mbox{$P$-a.s.}\]

\noindent This requires that the set of expected parameters converges to an $\epsilon$-neighborhood of $\mu^\infty$.
 Then, Theorem \ref{prop:Limit} holds as long as the set of learning rules $\mathcal{M}$ satisfies $\epsilon$-Uniform Convergence for some $\epsilon \leq \delta^\infty/(2K\xi)$. The rate results do not change. 

\paragraph{Approximate Common Certainty.}
Suppose that instead of imposing common certainty in $\mathscr{B}(\bz_n)$, as we have done in the main text, players have common $p$-belief in $\mathscr{B}(\bz_n)$. Formally, for any probability $p\in [0,1]$, player $i$, and set $\mathscr{B} \subseteq \Delta(\Theta)$, define 
$$B^{1,p}_i(\mathscr{B}) := \left\{t_i \in T^*_i \, : \,  \textstyle \marg_{\Theta} \kappa^*_i(t_i) \in \mathscr{B} \right\}.\footnote{This set has the same definition as $B^{1,1}_i$ from the main text. It is possible to relax the assumptions further, so that $B^{1,p}_i(\mathscr{B}) := \left\{t_i \in T^*_i \, : \, \sup_{\nu \in \mathscr{B}} d_P(\nu,\marg_{\Theta} \kappa^*_i(t_i)) \leq 1-p \right\}$, but this does not correspond to any standard definitions.}$$
For each $k>1$, and again for each player $i$, recursively define
\[B^{k,p}_i(\mathscr{B})= \left\{t_i \in T^*_i \, : \, \kappa^*_i(t_i)\left(\Theta \times \prod_{j \neq i} B_j^{k-1,p}(\mathscr{B})\right) \geq p\right\}.\] Then 
$T_i^{\mathscr{B},p} =  \bigcap_{k \geq 1} B_i^{k,p}(\mathscr{B})$ 
is the set of player $i$ types that have common $p$-belief in the event that all players' first-order beliefs belong to $\mathscr{B}$. 
There exists a $\overline{p}$ such that so long as players have common $p$-belief in the event that all players' first-order beliefs belong to $\mathscr{B}(\bz_n)$, where $p>\overline{p}$, then Theorem \ref{prop:Limit} holds as stated. Rate results similar to those in Section \ref{sec:finite} can also be obtained (see Online Appendix \ref{app:pbelief} for details).  Both extensions rely on boundedness of the payoff range.

\paragraph{Confidence Sets for Equilibrium.} The proposed approach can be paired with solution concepts besides rationalizability. For example, suppose we are interested in evaluating an analyst's confidence in predicting that the action profile $a\in A$ is part of a (pure-strategy) Bayesian Nash equilibrium. The analogous confidence set is $[\underline{p}^n(a), \overline{p}^n(a)]$, where the lower bound $\underline{p}^n(a)$ is the probability (over possible datasets $\bz_n$) that $a_i$ is a best reply to $a_{-i}$ for every player $i$ of any type $t_i \in T^{\mathscr{B}(\bz_n)}_i$. The upper bound $\overline{p}^n$ is the probability that there exists a belief-closed type space $(T_i, \kappa_i)_{i\in\mathcal{I}}$ where each $T_i \subseteq T^{\mathscr{B}(\bz_n)}_i$, and the strategy profile $\sigma$ with $\sigma_i(t_i) = a_i$ for all $i, t_i \in T_i$ is a Bayesian Nash equilibrium. Then, Theorem \ref{prop:Limit} holds with ``strict rationalizability" replaced with ``strict equilibrium" in the limiting game, and the rate results provided in Theorem \ref{prop:speedRobust} hold when $\delta^\infty$ is replaced with an analogous notion for the strictness of the equilibrium in the limiting game.

\section{Conclusion} 
Economists make predictions in incomplete information games based on models of unobservable beliefs. A large literature on the robustness of strategic predictions to the specification of agent beliefs provides guidance regarding whether these predictions should be trusted. These robustness notions tend to be qualitative---we learn whether the prediction is or isn't robust to perturbations in the agents' beliefs. Here I offer a different perspective, namely a quantitative metric for how robust the prediction is. The metric depends on the quantity of data that agents get to see. Predictions that hold given infinite quantities of data may not hold given large quantities of data, and those that hold given large quantities of data may not hold in environments where agents see only a few observations. Likewise, predictions that don't hold at the limit may nevertheless be plausible when agents' beliefs are coordinated by a small number of observations. The proposed framework provides a way of formalizing this, generating new comparative statics for how the analyst's confidence in a strategic prediction varies with primitives of the learning environment.

\pagebreak

\appendix
\small  
\linespread{1.1}

\begin{center}
\huge{Appendix}
\end{center}

\section{Proofs for Section \ref{sec:example}}

\subsection{Proof of Claim \ref{claim:Trade}} \label{app:exTrade}

Suppose that $f(x^S)=1$, so that the Seller's value good has a high value. (The proof follows along similar lines in the other case.) I will first show that $\underline{p}^n=0$ for every $n$.  Let $\pi$ be a point mass at $f$. An agent with this prior assigns probability 1 to $v=1$ no matter the outcome of the data. Hence, the degenerate distribution at 1 belongs to $\mathscr{B}(\bz_n)$ for every  $\bz_n$, so common certainty in  $v=1$ is consistent with Assumption \ref{ass:Example} with probability 1.\footnote{Here, and elsewhere in the proof, type $t_i$ has \emph{common certainty in $v=1$} if the type has common certainty in the event $\{f \mid f(x^S)=1\} \times (X\times \{0,1\})^\infty \times T^*_{-i}$.} But entering is not rationalizable for the Seller with this belief, implying $\underline{p}^n=0$ as desired.

To prove Parts (a) and (b) of the claim, which refer to the probability $\overline{p}^n$, I first show that entering is rationalizable for some type satisfying Assumption \ref{ass:Example} if and only if there exist $\tilde{f},\tilde{f}'\in \mathcal{F}$ that are consistent with the data, and which make conflicting predictions for the Seller's good $x^S$ (Lemma \ref{lemm:conditionRectangle}). I characterize the probability of this event in Lemma \ref{lemm:pnRectangle}, from which the comparative statics for $\overline{p}^n$ follow directly. 

\begin{lemma} 
Fix an arbitrary data set $\bz_n=\left\{\left(x_i,f(x_i)\right)\right\}_{i=1}^n$. Entering is rationalizable for the Seller with a belief satisfying Assumption \ref{ass:Example} if and only if there exist $\tilde{f},\tilde{f}' \in \mathcal{F}$ where
\begin{enumerate}
\item[(1)] $\tilde{f}(x_i)=\tilde{f}'(x_i)=f(x_i)$ for each observation $i=1, \dots, n$
\item[(2)] $\tilde{f}(x^S)=1$ while $\tilde{f}'(x^S)=0$ \\[-8mm]
\end{enumerate} \label{lemm:conditionRectangle}
\end{lemma}

\begin{proof} Suppose there exists a pair $\tilde{f},\tilde{f}'$ satisfying (1) and (2), and define  $\pi_{\tilde{f}},\pi_{\tilde{f}'} \in \Delta(\mathcal{F})$ to be point masses on $\tilde{f}$ and $\tilde{f}'$. Since these rules are consistent with the data by (1), the posterior beliefs updated to $\bz_n$ are likewise degenerate at $\tilde{f}$ and $\tilde{f}'$, and thus assign (respectively) probability 1 to $v=1$ and probability 1 to $v=0$. This implies that degenerate distributions at 1 and 0 both belong to $\mathscr{B}(\bz_n)$. Entering is rationalizable for the Seller who believes that $v=1$ with probability 1, and who believes with probability 1 that the Buyer believes $v=0$ with probability 1. This belief is consistent with common certainty that all players have first-order beliefs in $\mathscr{B}(\bz_n)$. 

Now suppose that no such pair $\tilde{f},\tilde{f}'$ exists, implying either that every $\tilde{f}\in \mathcal{F}$ consistent with the data predicts $f(x^S)=0$, or that every $\tilde{f}\in \mathcal{F}$ consistent with the data predicts $f(x^S)=1$. Then either $\mathscr{B}(\bz_n)$ is the singleton set consisting of a degenerate distribution at 1, or it is the singleton set consisting of a degenerate distribution at 0. If the former, the only type satisfying Assumption \ref{ass:Example} is the one with common certainty in $v=1$, and if the latter, the only type satisfying Assumption \ref{ass:Example} is the one with common certainty in $v=0$. Entering is not rationalizable for the Seller with either of these beliefs. 
\end{proof}

\begin{lemma} 
Suppose the true function is $f(x) = \mathbbm{1}(x\in R)$ where $R = [-\underline{r}_1,\overline{r}_1]\times [-\underline{r}_2,\overline{r}_2] \times \dots [-\underline{r}_m,\overline{r}_m]$ for a sequence of constants $\underline{r}_1,\overline{r}_1, \dots, \underline{r}_m,\overline{r}_m\in (0,1)$. Then 
\[\overline{p}^n(a_i)=1- \prod_{k=1}^m \left(1-\left(\frac12\right)^n [(2-\underline{r}_k)^n + (2-\overline{r}_k)^n- (2-(\underline{r}_k + \overline{r}_k))^n]\right).\]
\end{lemma}

\begin{proof} From Lemma \ref{lemm:conditionRectangle}, the probability $\overline{p}^n$ is equal to the measure of data sets $\bz_n$ given which there exist $\tilde{f},\tilde{f}'\in \mathcal{F}$ that are consistent with $\bz_n$, and which make conflicting predictions at the input $x^S$. The true classification rule $f$ is always consistent with the data, and predicts $f(x^S)=1$, so a pair of such rules exists if we can additionally find a rule $\tilde{f} \in \mathcal{F}$ consistent with the data that predicts $\tilde{f}(x^S)=0$. 

A necessary and sufficient condition for existence of such a rule is that there is some dimension $k$ on which either every observation $x_i$ satisfies $x_i^k < 0$, or every $x_i$ satisfies $ 0< x_j^k$. This allows some $\tilde{f}\in \mathcal{F}$ to be consistent with the data, but to predict 0 at the zero vector.

For each dimension $k$, the probability that there is at least one observation $x_i$ with $x_i^k \in [-\underline{r}_k,0)$ and at least one observation $x_j$ with $x_j^k\in (0,\overline{r}_k]$ is 
\[1-\left(\frac12\right)^n [(2-\underline{r}_k)^n + (2-\overline{r}_k)^n- (2-(\underline{r}_k + \overline{r}_k))^n].\]
Observe that attribute values are independent across dimensions. So the probability that for every dimension $k$, there is at least one observation $x_i^k \in [-\underline{r}_k,0)$ and at least one observation $x_j$ with $x_j^k\in (0,\overline{r}_k]$, is
\[\prod_{k=1}^m \left(1-\left(\frac12\right)^n [(2-\underline{r}_k)^n + (2-\overline{r}_k)^n- (2-(\underline{r}_k + \overline{r}_k))^n]\right).\]
The desired probability is the complement of this event, which yields the expression in the lemma. \end{proof}

The following functional form is used in the main text:
\begin{corollary} \label{lemm:pnRectangle}
In the special case in which the true function is $f(x) = \mathbbm{1}(x\in R)$ where $R = [-a,a]^m$ for some $a \in (0,1)$, then
$\overline{p}^n(a_i)=1- \left[1-\left(2 \left(\frac{2-a}{2}\right)^n - (1-a)^n\right)\right]^m.$
\end{corollary}

\subsection{Proof of Claim \ref{claim:Coord}} \label{app:exCoord}

I first demonstrate the following lemma, which characterizes the probabilities $\underline{p}^n$ and $\overline{p}^n$. 

\begin{lemma} For every $n\geq 1$,
\[\underline{p}^n = 1- \Phi \left(z_\alpha - \frac{\beta-1}{\sigma}\sqrt{\frac{n^2-1}{12}}\right)\]
while
\[\overline{p}^n = 1- \Phi \left(-z_\alpha - \frac{\beta-1}{\sigma}\sqrt{\frac{n^2-1}{12}}\right)\]
where $z_\alpha = -\Phi^{-1}(\alpha/2)$ with $\Phi$ denoting the CDF of the standard normal distribution.
\label{lemma:exCoord}
\end{lemma}

\noindent Since $\beta>1$ by assumption, both expressions are decreasing in $\sigma$ and increasing in $n$. Thus Claim \ref{claim:Coord} follows.

Towards this lemma, I first prove the following intermediate result:

\begin{lemma} \label{lemma:characterization} Write $T_i^{\mathscr{C}}$ for the set of player $i$ types with common certainty in the event that all players have first-order beliefs that assign probability 1 to $\mathscr{C}$. 
\begin{itemize}
\item[(a)] The strong policy is rationalizable for \emph{all} types $t_i \in T_i^{\mathscr{C}}$  if and only if $\mathscr{C} \subseteq [1,\infty)$.
\item[(b)] The strong policy is rationalizable for \emph{some} type $t_i \in T_i^{\mathscr{C}}$ if and only if  $\mathscr{C} \cap [1,\infty) \neq \emptyset$.
\end{itemize}
\end{lemma}

\begin{proof}
(a) $\mathscr{C} \subseteq [1,\infty)$ is a necessary condition, as otherwise the strong policy is not rationalizable for any type with common certainty in $\beta' \in \mathscr{C} \backslash [1,\infty)$. Suppose $\mathscr{C} \subseteq [1,\infty)$ and choose any $t_i \in T_i^{\mathscr{C}}$.   For each $\beta \in \mathscr{C}$,
$u_i(\mbox{strong}, \mbox{strong}, \beta) = -1$ while $u_i(\mbox{weak}, \mbox{strong}, \beta) =-\beta \leq -1$. So
\[\int u_i(\mbox{strong}, \mbox{strong}, \beta)t^1_i(\beta)d\beta = -1 \geq \int u_i(\mbox{weak}, \mbox{strong}, \beta)t^1_i(\beta)d\beta\]
where $t_i^1$ denotes the first-order belief of type $t_i$. Thus the family of sets $(R_1,R_2)$ with $R_1=R_2=\{\mbox{strong}\}$ are closed under best reply, and rationalizability of the strong policy follows.

 (b) Suppose $\mathscr{C} \cap [1,\infty)= \emptyset$. Then for every $\beta \in \mathscr{C}$,
 \[u_i(\mbox{strong}, \mbox{strong}, \beta) = -1\leq -\beta = u_i(\mbox{weak}, \mbox{strong}, \beta).\]
 So the strong policy is strictly dominated (and hence not rationalizable) for player $i$ given any type $t_i \in T_i^{\mathscr{C}}$. If instead $\mathscr{C} \cap [1,\infty)\neq \emptyset$, then the strong policy is rationalizable for any type with common certainty in some $\beta$ in this intersection. So the strong policy is rationalizable for at least one type $t_i \in T_i^{\mathscr{C}}$, as desired.
\end{proof}

I now prove Lemma \ref{lemma:exCoord}.

\begin{proof}
Using standard results for ordinary least-squares \citep{hastie}, the distribution of the OLS estimator $\hat{\beta}$ is
\[\hat{\beta} \sim \mathcal{N}\left(\beta, \frac{\sigma^2}{\frac1n \sum_{t=1}^n (t - \overline{t})^2}\right)\]
where $\overline{t}=\frac1n \sum_{t=1}^n t$. Since
\begin{align*}
\frac1n \sum_{t=1}^n (t - \overline{t})^2 & = \frac1n \left( \sum_{t=1}^n t^2 - 2\overline{t} \sum_{t=1}^n t  + \sum_{t=1}^n \overline{t}^2 \right) \\
&= \frac{(n+1)(2n+1)}{6} - \frac{(n+1)^2}{2} + \left(\frac{n+1}{2}\right)^2 = \frac{(n^2-1)}{12}
\end{align*}
we can simplify the variance of $\hat{\beta}$ to $\frac{12 \sigma^2}{n^2-1}$. The $(1-\alpha)$-confidence interval for $\beta$ given data $\bz_n$ is thus
\begin{equation} \label{eq:CI}
\mathscr{C}(\bz_n) = \left[\hat{\beta}(\bz_n) - z_\alpha \cdot \sigma \cdot \sqrt{\frac{12}{n^2-1}}, \hat{\beta}(\bz_n) + z_\alpha \cdot \sigma \cdot  \sqrt{\frac{12}{n^2-1}}\right]
\end{equation}
where $\beta(\bz_n)$ is the OLS estimate of $\beta$ given the data $\bz_n$, and $z_\alpha = -\Phi^{-1}(\alpha/2)$ is the critical value associated with the $(1-\alpha)$-confidence level.  The probability that the interval in (\ref{eq:CI}) is contained in $[1,\infty)$ is
\[\Pr \left(\hat{\beta}(\bz_n) > 1 + z_\alpha \cdot \sigma \cdot \sqrt{\frac{12}{n^2-1}} \right).\]
which is in turn equal to
\begin{equation} \label{eq:underpn}
1- \Phi \left(1.96 - \frac{\beta-1}{\sigma}\sqrt{\frac{n^2-1}{12}}\right).
\end{equation}
By Part (a) of Lemma \ref{lemma:characterization}, $\underline{p}^n$ is equal to (\ref{eq:underpn}), delivering the first part of the lemma. 

The probability that the interval in (\ref{eq:CI}) has nonempty intersection with $[1,\infty)$ is given by 
\[\Pr \left(\hat{\beta}(\bz_n) > 1 - z_\alpha \cdot \sigma \cdot \sqrt{\frac{12}{n^2-1}} \right)\]
which is equal to
\begin{equation} \label{eq:overpn}
1- \Phi \left(-z_\alpha - \frac{\beta-1}{\sigma}\sqrt{\frac{n^2-1}{12}}\right)
\end{equation}
By Part (b) of Lemma \ref{lemma:characterization}, $\overline{p}^n$ is equal to (\ref{eq:overpn}), concluding the proof. 
\end{proof}

\section{Proofs for Main Results (Sections \ref{sec:Asymptotic} and \ref{sec:finite})}
\subsection{Proof of Theorem \ref{prop:Limit} Part (a) }

Recall that $\Theta$ and $U$ are both endowed with the sup-norm, and the map $g: \Theta \rightarrow U$ has Lipschitz constant $K$. The set of probability measures $\Delta(\Theta)$ is endowed with the Prokhorov metric $d_P$. The \emph{Wasserstein distance} on $\Delta(\Theta)$ is
\[d_W(\nu,\nu') = \sup \left\{ \int h d\nu- \int h d\nu' \, : \, \|h\|_L \leq 1 \right\}\]
where $\|h\|_L$ is the Lipschitz constant of the function $h: \Theta \rightarrow \mathbb{R}$.

\begin{lemma}  \label{lemm:Lipschitz} Fix any player $i$, action $a_i \in A_i$, mixed strategy $\alpha_i \in \Delta(A_i)$, and set $R_{-i} \subseteq A_{-i}$. Let $a_{-i}(\theta): \Theta \rightarrow \Delta(A_{-i})$  be any function satisfying
\[a_{-i}(\theta) \in \argmax_{a_{-i}\in R_{-i}} (u_i(a_i,a_{-i},\theta)- u_i(\alpha_i,a_{-i},\theta)) \quad \forall \theta \in \Theta\]
and define $h: \Theta \rightarrow \mathbb{R}$ by
\[h(\theta) = u_i(a_i, a_{-i}(\theta), \theta) - u_i(\alpha_i, a_{-i}(\theta),\theta).\]
Then, the function $h$ is Lipschitz continuous with Lipschitz constant $2K$.
\end{lemma}

\begin{proof} Choose any $\theta, \theta'\in \Theta$, and without loss of generality, suppose $h(\theta)\geq h(\theta')$. Then
\begin{align*}
\vert h(\theta) - h(\theta') \vert & = \vert (u(a_i,a_{-i}(\theta),\theta) - u(\alpha_i, a_{-i}(\theta),\theta)) - \\
& \quad \quad \quad \quad \quad \quad \quad (u(a_i,a_{-i}(\theta'),\theta') - u(\alpha_i, a_{-i}(\theta'),\theta')) \vert \\
& \leq \vert (u(a_i,a_{-i}(\theta),\theta) - u(\alpha_i, a_{-i}(\theta),\theta)) - \\
& \quad \quad \quad \quad \quad \quad \quad (u(a_i,a_{-i}(\theta),\theta') - u(\alpha_i, a_{-i}(\theta),\theta')) \vert \\
& \leq \vert u(a_i,a_{-i}(\theta),\theta) - u(a_i,a_{-i}(\theta),\theta')  \vert +  \\
& \quad \quad \quad \quad \quad \quad \quad \vert u(\alpha_i, a_{-i}(\theta),\theta) - u(\alpha_i, a_{-i}(\theta),\theta') \vert \\
& \leq 2 \| g(\theta) - g(\theta')\|_\infty \leq 2K \| \theta - \theta' \|_\infty
\end{align*}
using in the final inequality that $g: \Theta \rightarrow U$ has Lipschitz constant $K$. \end{proof}

Below, let $F^\epsilon$ denote the $\epsilon$-neighborhood of the set $F$.

\begin{lemma} \label{lemm:eps}
Suppose $a_i$ is $\delta$-strictly rationalizable for player $i$ of type $t_i^\infty$, where $\delta>0$. Let $\mathscr{B}$ be any subset of $\{\mu^\infty\}^{ \delta/(2K\xi)}$, where $K$ is the Lipschitz constant of $g: \Theta \rightarrow U$, and $\xi$ is as defined in (\ref{eq:xi}). Then, $a_i$ is rationalizable for all types $t_i \in T_i^{\mathscr{B}}$.\footnote{\citet{faingold} demonstrate a similar result for finite state spaces $\Theta$ (see their Proposition 2). I use ideas from their proof here, but consider a more general environment, replacing finiteness of $\Theta$ with Lipschitz continuity on $g:\Theta \rightarrow U$. }
\end{lemma}

\begin{proof}  Fix $\epsilon >0$, and consider an arbitrary set $\mathscr{B} \subseteq \{\mu^\infty\}^\epsilon$.  I will show that $a_i$ is rationalizable for all types $t_i \in T_i^{\mathscr{B}}$ when $\epsilon$ is sufficiently small.

To show this, I use Proposition 1 from \citet{faingold}:\footnote{Proposition 1 from \citet{faingold} characterizes $\gamma$-rationalizability for arbitrary $\gamma \in \mathbb{R}$. For the purposes of this proof, it is sufficient to set $\gamma=0$.}

\begin{proposition}[\citet{faingold}] \label{prop:CDFX} For each $k\geq 1$, player $i \in \mathcal{I}$, type $t_i \in T_i$, and action $a_i \in A_i$, we have $a_i \in S^k_i[t_i]$ if and only if for each $\alpha_i \in \Delta(A_i \backslash \{a_i\})$, there exists a measurable $\sigma_{-i}: \Theta \times T_{-i} \rightarrow \Delta(A_{-i})$ with
\[\supp \sigma_{-i}(\theta, t_{-i}) \subseteq S_{-i}^{k-1}[t_{-i}] \quad \forall (\theta, t_{-i}) \in \Theta \times T_{-i}\]
such that 
\[\int_{\Theta \times T_{-i}} [u_i(a_i, \sigma_{-i}(\theta, t_{-i}), \theta) - u_i(\alpha_i, \sigma_{-i}(\theta,t_{-i}),\theta] t_i[d\theta \times dt_{-i}] \geq 0\]
\end{proposition}

By assumption, there is a $\delta \in \mathbb{R}_{++}$ such that $a_i$ is $\delta$-strictly rationalizable for player $i$ of type $t_i^\infty$. This implies that there exists a family of sets $(R_j)_{j \in \mathcal{I}} \subseteq \prod_{j \in \mathcal{I}} A_j$, where $a_i \in R_i$, and for every $a_j \in R_j$ there exists a $\sigma^\infty_{-j}: \Theta \rightarrow \Delta(A_{-j})$ satisfying
 \[\supp \sigma^\infty_{-j}(\theta) \subseteq R_{-j} \quad \forall \theta \in \Theta\]
 and  
\begin{equation} \label{eq:Infty}
\int_{\Theta} u_j(a_j, \sigma^\infty_{-j}(\theta), \theta) d\mu^\infty  - \int_{\Theta} u_j(a'_j, \sigma^\infty_{-j}(\theta), \theta) d\mu^\infty  \geq \delta \quad \forall a_j'\neq a_j
\end{equation}

I will show that for each $k\geq 1$, player $j$, type $t_j \in T_j^{\mathscr{B}}$, action $a_j \in R_j$, and mixed strategy $\alpha_j \in \Delta(A_j \backslash \{a_j\})$, there exists a measurable $\sigma_{-j}: \Theta \times T_{-j}^{\mathscr{B}} \rightarrow \Delta(A_{-j})$ with 
\[\supp \sigma_{-j}(\theta, t_{-j}) \subseteq R_{-j} \quad \forall (\theta, t_{-j}) \in \Theta \times T^{\mathscr{B}}_{-j}\]
and
\begin{equation} \label{eq:BR}
\int_{\Theta \times T^{\mathscr{B}}_{-j}} [u_j(a_j, \sigma_{-j}(\theta, t_{-j}), \theta) - u_j(\alpha_j, \sigma_{-j}(\theta,t_{-j}),\theta] t_j[d\theta \times dt_{-j}] \geq 0.
\end{equation}
Since $a_i \in R_i$ by design, it follows from Proposition \ref{prop:CDFX} that for any type $t_i \in T_i^{\mathscr{B}}$, the action $a_i \in S^k_i[t_i]$ for every $k$, and hence $a_i \in S^\infty_i[t_i]$, as desired.

Fix an arbitrary player $j$, $a_j \in R_j$, type $t_j \in T_j^{\mathscr{B}}$, and $\alpha_j \in \Delta(A_j \backslash \{a_j\})$. Define $a_{-j}: \Theta \rightarrow A_{-j}$ to satisfy
\[a_{-j}(\theta) \in \argmax_{a_{-j}\in R_{-j}} (u_j(a_j,a_{-j},\theta)- u_j(\alpha_j,a_{-j},\theta)) \quad \forall \theta \in \Theta\]
and define $\sigma_{-j}: \Theta \times T_{-j}^{\mathscr{B}} \rightarrow \Delta(A_{-j})$ so that each $\sigma_{-j}(\theta, t_{-j})$ is a point mass at $a_{-i}(\theta)$. Then by definition
\[\supp \sigma_{-j}(\theta,t_{-j}) \subseteq R_{-j} \quad \forall (\theta, t_{-j}) \in \Theta \times T_{-j}^{\mathscr{B}}.\]
Further define 
\[h(\theta) := u_j(a_j, a_{-j}(\theta), \theta) - u_j(\alpha_j, a_{-j}(\theta),\theta) \quad \forall \theta \in \Theta.\]
For notational ease, write $\nu \in \Delta(\Theta)$ for the first-order belief of type $t_j$.  Then 
\begin{align*}
\int_{\Theta \times T^{\mathscr{B}}_{-j}}  &  u_j(a_j, \sigma_{-j}(\theta, t_{-j}), \theta) t_j[d\theta \times dt_{-j}] - \int_{\Theta \times T^{\mathscr{B}}_{-j}} u_j(\alpha_j, \sigma_{-j}(\theta, t_{-j}), \theta) t_j[d\theta \times dt_{-j}] \\ 
& = \int_{\Theta} u_j(a_j, a_{-j}(\theta), \theta) \nu[d\theta] - \int_{\Theta} u_j(\alpha_j, a_{-j}(\theta), \theta) \nu[d\theta]  = \int_\Theta h(\theta)\nu [d\theta]
\end{align*}
so the desired condition in (\ref{eq:BR}) follows if we can show that $\int_\Theta h(\theta)\nu[d\theta] \geq 0.$

By Lemma \ref{lemm:Lipschitz}, the function $h: \Theta \rightarrow \mathbb{R}$ has Lipschitz constant $2K$, so
\[\left\vert \int_\Theta h(\theta) d\nu - \int_\Theta h(\theta) d\mu^\infty \right\vert \leq 2K \cdot d_W(\nu,\mu^\infty)\]
where $d_W$ is the Wasserstein distance on $\Delta(\Theta)$. This implies
\[ \int_\Theta h(\theta) d\nu \geq  \int_\Theta h(\theta) d\mu^\infty  - 2K \cdot d_W(\nu,\mu^\infty).\]
 Applying Theorem 2 in \citet{metric},
$d_W(\nu,\mu^\infty) \leq \xi \cdot d_P(\nu,\mu^\infty)$, 
where $d_P$ is the Prokhorov distance on $\Delta(\Theta)$ and $\xi$ is as defined in (\ref{eq:xi}). So 
\begin{equation} \label{eq:tosimplify} \int_\Theta h(\theta) d\nu \geq  \int_\Theta h(\theta) d\mu^\infty  - 2K \xi \cdot d_P(\nu,\mu^\infty)
\end{equation}
It follows from the inequality in (\ref{eq:Infty}) that
\begin{align*}
\int_\Theta h(\theta) d\mu^\infty = \int_{\Theta} u_j(a_j, \sigma^\infty_{-j}(\theta), \theta) d\mu^\infty  - \int_{\Theta} u_j(\alpha_j, \sigma^\infty_{-j}(\theta), \theta) d\mu^\infty \geq \delta,
\end{align*}
so (\ref{eq:tosimplify}) implies
\[\int_\Theta h(\theta) d\nu \geq  \delta  - 2K \xi \cdot d_P(\nu,\mu^\infty).\]
Finally, by assumption that $t_j \in T_j^{\mathscr{B}}$ for some $\mathscr{B}\subseteq \{\mu^\infty\}^\epsilon$, the Prokhorov distance between the first-order belief of type $t_j$ and the limiting belief $\nu^\infty$ is $d_P(\nu,\mu^\infty) \leq \epsilon.$
 So
\[\int_\Theta h(\theta)d\nu \geq \delta - 2K\xi \epsilon.\]
It follows that $\epsilon \leq \delta/(2K\xi)$ is a sufficient condition for the constructed $\sigma_{-j}$ to satisfy the desired condition in (\ref{eq:BR}). 
\end{proof}

Since $a_i$ is strictly rationalizable for type $t_i^\infty$ (by assumption), there exists a $\delta\in \mathbb{R}_{++}$ for which $a_i$ is $\delta$-strictly rationalizable.  Assumption \ref{ass:UniformConv} implies that
\[\lim_{n \rightarrow \infty} P^n\left(\left\{\bz_n \mid \sup_{\mu \in \mathcal{M}} d_P(\mu(\bz_n), \mu^\infty) \leq \epsilon \right\}\right) =0 \quad \quad \forall \epsilon >0.
\]
which further implies
\begin{equation} \label{eq:convnbhd}
\lim_{n \rightarrow \infty} P^n\left(\left\{\bz_n \mid \mathscr{B}(\bz_n) \subseteq \{\theta^\infty\}^{\delta/2K\xi} \right\}\right) = 0
\end{equation}
By Lemma \ref{lemm:eps},
\[\underline{p}^n(a_i) \geq P^n\left(\left\{\bz_n \mid \mathscr{B}(\bz_n) \subseteq \{\theta^\infty\}^{\delta/2K\xi}\right\}\right) \quad \quad \forall n\geq 1\]
so from (\ref{eq:convnbhd}) we can directly conclude that $\underline{p}^n(a_i)\rightarrow 1$. Theorem \ref{prop:Limit} Part (a) follows.

\subsection{Proof of Theorem \ref{prop:Limit} Part (b)}

I begin by reviewing definitions from \citet{faingold} that will be used in the proof. For each player $i$, let $X_i^0=\Theta$, and recursively for $k\geq 1$, define $X_i^k = \Theta \times \prod_{j \neq i} \Delta(X_j^{k-1})$.\footnote{The sets $X^k$ defined in Section \ref{sec:preliminaries} can be identified with the sets $X^k$ defined in this way.} The  space of \emph{$k$-th order beliefs} for player $i$ is defined $T_i^k := \Delta(X_i^{k-1})$, noting that each $T_i^k = \Delta(\Theta \times T_{-i}^{k-1})$. The \emph{uniform-weak} metric on the universal type space $T^*_i$ is
\[d^{UW}_i(s_i,t_i) = \sup_{k\geq 1} d^k_i(s_i, t_i) \quad \forall s_i, t_i \in T^*_i\]
where $d^0$ is the supremum norm on $\Theta$ and recursively for $k\geq 1$, $d_i^k$ is the Prokhorov distance on $\Delta(\Theta \times T_{-i}^{k-1})$ induced by the metric $\max\{d^0,d_{-i}^{k-1}\}$ on $\Theta \times T_{-i}^{k-1}$.\footnote{This definition is slightly modified from \citet{faingold}, where $d^0$ was the discrete metric on $\Theta$. The change reflects the difference that $\Theta$ was taken to be a finite set in \citet{faingold}, while it is a compact and convex subset of Euclidean space here.} The \emph{uniform-weak topology} on the universal type space is the metric topology induced by $d^{UW}_i$.

\begin{lemma} \label{lemm:UWdistance} Let $\mathscr{B}$ be a subset of $\{\mu^\infty\}^\epsilon$, and choose any $s_i \in T_i^{\mathscr{B}}$. Then, $d^{UW}_i(t^\infty_i,s_i) \leq \epsilon$. 
\end{lemma}
\begin{proof} 
For simplicity of notation, write $t_i$ for $t_i^\infty$. It will be useful to define
\[T_i^{\mathscr{B},k} = \left\{s_i^k \in T_i^k \mid s_i \in T_i^{\mathscr{B}}\right\} \]
for the set of all $k$-th order beliefs that are consistent with some type $s_i \in T_i^{\mathscr{B}}$.\footnote{Here and elsewhere, $t^k_i$ denotes the $k-th$ order belief of type $t_i$.} I will show that
\begin{equation}\label{eq:toshow}
d_P\left(T_i^{\mathscr{B},k}, t_i^k\right) := \sup_{s_i^k \in T_i^{\mathscr{B},k}
}  d_P(s_i^k, t_i^k) \leq \epsilon \quad \forall k\geq 1 \end{equation}
from which the desired lemma directly follows.

By construction, $T_i^{\mathscr{B},1} = \mathscr{B}$, so  the assumption  $\mathscr{B} \subseteq \{\mu^\infty\}^\epsilon$ immediately implies (\ref{eq:toshow}) for $k=1$. Proceed by induction. Suppose $d_P\left(T_i^{\mathscr{B},k}, t_i^k\right) \leq \epsilon$, and consider any measurable set $E \subseteq T^k$. If $t^{k}_i\in E$, then $t^{k+1}_i(E) =1$ by definition of $t_i$. Also,
\begin{align*}
s^{k+1}_i \left(E^{\epsilon}\right) \geq s^{k+1}_i \left(\{t_i^k\}^{\epsilon}\right) \geq s^{k+1}_i(T_i^{\mathscr{B},k}) = 1
\end{align*}
where the second inequality follows from the inductive hypothesis, and the final inequality follows by assumption that $s_i \in T_i^{\mathscr{B}}$. So
\begin{equation} 
\label{prokh} t_i^{k+1} (E) \leq s_i^{k+1}\left(E^{\epsilon} \right) + \epsilon.
\end{equation}
If $t^k_i \notin E$, then $t_i^{k+1}(E)=0$ (again by definition of $t_i$), so (\ref{prokh}) follows trivially. Thus 
\[d^{k+1}_i(t_i,s_i) = \inf \{\delta \mid t_i^{k+1}(E) \leq s_i^{k+1}(E^\delta) + \delta \quad \forall \mbox{ measurable } E\subseteq T_i^k \} \leq \epsilon\] and so $d^{UW}_i(t_i, s_i) = \sup_{k\geq 1} d^k_i(t_i,s_i) \leq \epsilon$ as desired.
\end{proof}

Lemma \ref{lemm:UWdistance} implies the subsequent corollary.

\begin{corollary} \label{lemm:UW} Suppose Assumption \ref{ass:UniformConv} holds. Consider any sequence $\bz \in \mathcal{Z}^\infty$ satisfying
\begin{equation} \label{eq:conv}
\lim_{n \rightarrow \infty} \sup_{\mu \in \mathcal{M}} d_P(\mu(\bz_n), \mu^\infty))=0
\end{equation}
and choose any sequence of types $(s_i^n)_{n=1}^\infty$ with $s_i^n \in T_i^{\mathscr{B}(\bz_n)}$ for each $n\geq 1$. Then 
\[\lim_{n \rightarrow \infty} d^{UW}_i(t_i^\infty,s_i^n) = 0.\]
\end{corollary}

Now we will complete the proof of Theorem \ref{prop:Limit} Part (b). By Assumption \ref{ass:UniformConv}, there is a set $\mathcal{Z}^* \subseteq \mathcal{Z}^\infty$ of $P$-measure 1 such that
\begin{equation} \label{eq:Z*}
\lim_{n \rightarrow \infty} \sup_{\mu \in \mathcal{M}} d_P(\mu(\bz_n),\mu^\infty) = 0 \quad \forall \bz \in \mathcal{Z}^*
\end{equation}
Suppose towards contradiction that $\overline{p}^n(a_i) \nrightarrow 0$. Then, there is a set $\widehat{\mathcal{Z}} \subseteq \mathcal{Z}^\infty$ with strictly positive $P$-measure such that for every $\bz \in \widehat{\mathcal{Z}}$, there is a sequence of types $(t_i^n(\bz))_{n=1}^\infty$ where $t_i^n(\bz) \in T_i^{\mathscr{B}(\bz_n)}$ for every $n \geq 1$, and $a_i \in S_i^\infty[t_i^n(\bz)]$ for all $n$ sufficiently large.

But since $\mathcal{Z}^*$ has $P$-measure 1, it must be that $\widehat{\mathcal{Z}} \cap \mathcal{Z}^* \neq \emptyset$. Choose any $\bz$ from this intersection. Then, Lemma \ref{lemm:UW} and the display in (\ref{eq:Z*}) imply that $t_i^n(\bz) \rightarrow t_i^\infty$ in the uniform-weak topology. But rationalizability is upper hemi-continuous in the uniform-weak topology (Theorem 1, \citet{faingold}). So $a_i \notin S_i^\infty[t_i^\infty]$ implies $a_i \notin S_i^\infty[t_i^n(\bz)]$ for infinitely many $n$, a contradiction.

\subsection{Proof of Proposition \ref{prop:speedRobust}} 

By assumption, $a_i$ is strictly rationalizable for type $t_i^\infty$, so $\delta^\infty > 0$. Applying Lemma \ref{lemm:eps},
\begin{align*}
\underline{p}^n(a_i) & \geq P^n(\{\bz_n \, : \, \mathscr{B}(\bz_n)\subseteq \{\mu^\infty\}^{\delta^\infty/(2K\xi)}\}) \\
& =  P^n \left(\left\{ \bz_n \, : \, \sup_{\mu \in \mathcal{M}} d_P(\mu(\bz_n),\mu^\infty) \leq \delta^\infty/(2K\xi) \} \right\} \right) \\
& \geq 1- \frac{2K\xi}{\delta^\infty}  \mathbb{E} \left( \sup_{\mu \in \mathcal{M}} d_P(\mu(Z^n),\mu^\infty)  \right) 
\end{align*}
using Markov's inequality in the final line.

\subsection{Proof of Proposition \ref{corr:ShrinkSet}}

Suppose $a_i$ is strictly rationalizable for player $i$ in the complete information game $\theta^\infty$, and let $\delta^\infty$ be as defined in (\ref{eq:deltainf}). Then, there exists a family of sets $(R_j)_{j \in \mathcal{I}}$ with $a_i \in R_i$, where for each player $j$ and action $a_j \in R_j$, there is a mixed strategy $\sigma_{-j} \in \Delta(A_{-j})$ satisfying $\sigma_{-j}[R_{-j}]=1$, and
\[u_i(a_i,\sigma_{-j},\theta^\infty) - u_i(a'_i,\sigma_{-j},\theta^\infty) \geq \delta \quad \forall a_i'\neq a_i.\]

Now consider an arbitrary set $\mathscr{C} \subseteq \{\theta^\infty\}^\epsilon$ and a type $t_i$ with common certainty in the event that every player's first-order belief assigns probability 1 to $\mathscr{C}$. Write $\nu \in \Delta(\Theta)$ for the first-order belief of type $t_i$. For any action $a_j' \neq a_j$,
\begin{align*}
\int u_j (a_j, \sigma_{-j}, \theta) d\nu -  \int u_j(a_j', \sigma_j, \theta)d\nu & = \int u_j (a_j, \sigma_{-j}, \theta) d\nu - \int u_j (a_j, \sigma_{-j}, \theta) d\mu^\infty \\
& \quad \quad +  \int u_j (a_j, \sigma_{-j}, \theta) d\mu^\infty -  \int u_j (a'_j, \sigma_{-j}, \theta) d\mu^\infty \\
& \quad \quad \quad \quad +   \int u_j (a'_j, \sigma_{-j}, \theta) d\mu^\infty - \int u_j(a_j', \sigma_j, \theta)d\nu \\
& \geq  \int u_j (a_j, \sigma_{-j}, \theta) d\mu^\infty -  \int u_j (a'_j, \sigma_{-j}, \theta) d\mu^\infty \\
& \quad \quad  - \left \vert \int u_j (a_j, \sigma_{-j}, \theta) d\nu - \int u_j (a_j, \sigma_{-j}, \theta) d\mu^\infty \right \vert \\
& \quad \quad \quad \quad -  \left\vert \int u_j (a'_j, \sigma_{-j}, \theta) d\nu - \int u_j(a_j', \sigma_j, \theta)d\mu^\infty \right\vert \\
&  \geq \delta -  2 K \cdot d_P(\nu,\mu^\infty) \geq \delta - 2K\epsilon
\end{align*}
using in the penultimate inequality that $g:\Theta\rightarrow U$ has Lipschitz constant $K$. Since this bound on the payoff difference holds across all actions $a_j'\neq a_j$, the action $a_j$ is a  best reply to belief $\nu$ whenever $\epsilon \leq \delta / (2K)$.

This allows us to construct the lower bound
\begin{align*}
\underline{p}^n(a_i) & \geq Q^n\left(\left\{\bz_n \, : \, \mathscr{C}(\bz_n)\subseteq \{\theta^\infty\}^{\delta^\infty/(2K)}\right\}\right) \\
& =  Q^n \left(\left\{ \bz_n \, : \, \sup_{\theta\in \mathscr{C}(\bz_n)} \left\| \theta - \theta^\infty \right\|_\infty \leq \delta^\infty/(2K) \right\} \right) \\
& \geq 1- \frac{2K}{\delta^\infty}  \mathbb{E} \left(  \sup_{\theta\in \mathscr{C}(\bz_n)} \left\| \theta - \theta^\infty \right\|_\infty\right) 
\end{align*}
using Markov's inequality in the final line.

\clearpage

\renewcommand{\thesection}{\Alph{section}}
\renewcommand\thesubsection{\thesection.\arabic{subsection}}

\setcounter{section}{14}

\section{For Online Publication}

\subsection{Proof of Claim \ref{claim:Discontinuity}}

Fix an arbitrary $(\pi,q) \in (0,1)\times (1/2,1)$. Given data $\bz_n$, the posterior belief $\mu_{\pi,q}(\bz_n)$ assigns probability
\begin{equation} \label{eq:post1}
\hat{v}(\pi,q,\bz_n) := 1 / \left(1+ \frac{1-\pi}{\pi} \left(\frac{1-q}{q}\right)^{n(2\overline{\bz}_n-1)}\right)
\end{equation}
to $v=1$, where $\overline{\bz}_n=\frac1n \sum_{n'=1}^n z_{n'}$ denotes the average realization in the sequence $\bz_n$. 

Suppose without loss that $v=1$, and let $q^* \in (1/2,1)$ be the true frequency of $z=1$.  By the strong Law of Large numbers, there is a measure 1 set of sequences $\mathcal{Z}^*$ satisfying $\lim_{n \rightarrow \infty} (\frac1n \sum_{n'=1}^n z_{n'})= q^*$ for every $\bz = (z_1, z_2, \dots) \in \mathcal{Z}^*$. The expression in (\ref{eq:post1}) converges to 1 on this set for every $(\pi,q) \in (0,1) \times (1/2,1)$.  So Assumption \ref{ass:limitbelief} is satisfied, and the limiting belief $\mu^\infty$ assigns probability 1 to $v=1$. Since entering is not rationalizable for the Seller given common certainty in the event in that all players assign probability 1 to $v=1$, it follows that $\underline{p}(\infty)=\overline{p}(\infty)=0$. 

I show next that the probability $\overline{p}(n)$ converges to 1 as $n \rightarrow \infty$. 
Fix an arbitrary $n$, and define $\mathcal{Z}^\dag_n = \{\bz_n \mid \overline{\bold{z}}_n>1/2\}$ to be the set of length-$n$ sequences with majority realizations of $z=1$. For every $\bz_n \in \mathcal{Z}^\dag_n$, the expression $((1-q)/q)^{n(2\bold{z}_n-1)}$ is bounded between 1/2 and 1 on the domain $q\in (1/2,1)$, while the image of $(1-\pi)/\pi$ is all of $\mathbb{R}_+$. Thus, the display in (\ref{eq:post1}) ranges from zero to 1; that is,
\[\left\{\hat{v}(\pi,q,\bz_n) : \pi\in (0,1), q\in (1/2,1)\right\} = (0,1) \quad \forall \bz_n \in \mathcal{Z}^\dag_n.\]
 It follows that for every $\bz_n \in \mathcal{Z}^\dag_n$, there exist pairs $(\pi,q),(\pi',q') \in (0,1)\times (1/2,1)$ satisfying $\hat{v}(\pi,q,\bz_n)<p<\hat{v}(\pi',q',\bz_n)$. Entering is rationalizable for the Seller with a type that assigns probability $\hat{v}(\pi',q',\bz_n)$ to the high value, and which assigns probability 1 to the Buyer assigning probability $\hat{v}(\pi,q,\bz_n)$ to the high value. So entering is weakly $\mathscr{B}(\bz_n)$-rationalizable for every $\bz_n \in \mathcal{Z}^\dag_n$, implying $\overline{p}^n(a_i) \geq P^n(\mathcal{Z}^\dag_n)$.
 
 Again by the law of large numbers,  the measure of datasets with majority realizations of $z=1$ converges to 1 as $n \rightarrow \infty$; that is,
$
P^n\left(\mathcal{Z}^\dag_n\right)\rightarrow 1
$. So  $\lim_{n\rightarrow \infty} \overline{p}^n(a_i)=1$, as desired.

\subsection{Proof of Corollary \ref{corr:normal}}

First observe that $\delta^\infty=\beta-1$, since the action \emph{Strong} is $\delta$-strictly rationalizable for every $\delta < \beta-1$ and not for any $\delta \geq \beta-1$. It remains to determine $\mathbb{E}\left[\sup_{\theta' \in \mathscr{C}(Z^n)}\| \theta' - \theta^\infty\|_\infty\right]$. Write $\overline{Z}_n$ for the (random) empirical mean of $n$ signal realizations, and $\hat{\beta}_x(\bz_n)$ for the expectation of $\beta$ given signals $\bz_n$ and prior $\beta \sim \mathcal{N}(x,1)$. Then, using standard formulas for updating to Gaussian signals:
\begin{align*}
\mathbb{E}\left(\sup_{x\in [-\eta,\eta]} \vert \beta - \hat{\beta}_x(Z^n)\vert \right)
& = \mathbb{E} \left[ \max_{x \in [-\eta,\eta]} \left(\left\vert \beta - \frac{x + n\overline{Z}_n}{n+1} \right\vert \right)\right] 
\end{align*}
We can further bound the RHS as follows:
\begin{align*}
\mathbb{E} \left[ \max_{x \in [-\eta,\eta]} \left(\left\vert \beta - \frac{x + n\overline{Z}_n}{n+1} \right\vert \right)\right] 
 & \leq \mathbb{E} \left(\left\vert \beta -\frac{n\overline{Z}_n}{n+1}  \right\vert \right)+\max_{x \in [-\eta,\eta]} \left\vert \frac{x}{n+1}\right\vert \\
& = \mathbb{E} \left(\left\vert \beta -\frac{n\overline{Z}_n}{n+1} \right\vert \right)+ \eta/(n+1) \\
& \leq \mathbb{E} \left(\left\vert \beta -\overline{Z}_n \right\vert \right) + \mathbb{E}\left(\frac{\overline{Z}_n}{n+1}\right)+ \eta/(n+1) \\
& = \sqrt{\frac{2}{n\pi}} + \frac{\beta+\eta}{n+1}
\end{align*}
using in the final line the expected absolute deviation of the empirical mean of $n$ observations from a Gaussian distribution \citep{Geary}. Finally, the map $g: \Theta \rightarrow U$ has Lipschitz constant 1. Applying Proposition \ref{corr:ShrinkSet}, we have the desired bound.

\subsection{Proof of Corollary \ref{corr:Sanov}}

Fix arbitrary $\underline{\pi}, \overline{\pi},\underline{q},\overline{q}$ satisfying $0<\underline{\pi}<\overline{\pi}<1$ and $1/2<\overline{q}<\overline{q}<1$, and let $\mathcal{M}$ be the set of learning rules identified with $(\pi,q) \in  [\underline{\pi}, \overline{\pi}] \times [\underline{q},\overline{q}]$.  Entering is rationalizable for a Seller with common certainty that all players have first-order beliefs in $\mathscr{B}(\bz_n)$ if and only if there exist $\pi,\pi'\in [\underline{\pi},\overline{\pi}]$ and $q,q'\in [\underline{q},\overline{q}]$ satisfying
\begin{equation} \label{eq:piq}
\hat{v}(\pi,q,\bz_n) < p < \hat{v}(\pi',q',\bz_n).
\end{equation}
where $\hat{v}(\pi,q,\bz_n)$ is as defined in (\ref{eq:post1}). Let $\mathbb{Z}^*_n$ denote the set of all sequences $\bz_n$ satisfying (\ref{eq:piq}).
 
 Since the state space is binary, each empirical measure $\widehat{Q}(\bz_n) \in \Delta(\{0,1\})$ can be identified with its average realization $\overline{\bold{z}}_n$, which is also the probability assigned to $z=1$. The KL-divergence between $\widehat{Q}(\bz_n)$ and the actual signal-generating distribution $Q=(q^*,1-q^*)$ is 
\[D_{KL}(\widehat{Q}(\bz_n) \mid Q) = q^*\log\left(\frac{q^*}{\overline{\bold{z}}_n}\right) + (1-q^*)\log\left(\frac{1-q^*}{1-\overline{\bold{z}}_n}\right)\]
and this expression is monotonically increasing in $\vert \overline{\bold{z}}_n-q^*\vert$. Thus, to minimize the KL-divergence, we seek the value of $\overline{\bold{z}}_n$ closest to $q^*$ for which (\ref{eq:piq}) is satisfied.  

Suppose $\overline{\bold{z}}_n>1/2$.  By assumption, $\overline{\pi}>p$ and $\overline{q}>1/2$,  so $\hat{v}(\overline{\pi},\overline{q}, \bz_n)>p$. It remains to determine when $\hat{v}(\pi,q, \bz_n)<p$ is  satisfied for some other $(\pi,q)\in \mathcal{M}$. Since $\hat{v}(\pi,q, \bz_n)$ is monotonically increasing in both $\pi$ and $q$ for sequences $\bz_n$ satisfying $\overline{\bz}_n>1/2$ (and on the given domain for ($\pi,q$)), a necessary and sufficient condition is $\hat{v}(\underline{\pi},\underline{q}, \bz_n)<p$. Using (\ref{eq:post1}), this inequality requires
\begin{align*}
1 / \left(1+ \frac{1-\underline{\pi}}{\underline{\pi}} \left(\frac{1-\underline{q}}{\underline{q}}\right)^{n(2\overline{\bz}_n-1)}\right) <p
\end{align*}
which can be rewritten
\[\overline{\bold{z}}_n \leq \frac12 \left(1 + \frac1n \log_{(1-\underline{q})/\underline{q}}\left(\frac{\underline{\pi}}{1-\underline{\pi}} \cdot \frac{1-p}{p}\right)\right):= \overline{\bold{z}}_n^*.\]
Since $\overline{\bold{z}}_n^* \cdot n$ need not be an integer, the distribution $(\overline{\bold{z}}_n^*,1-\overline{\bold{z}}_n^*)$ may not be achievable by any empirical measure $\widehat{Q}_n$ for finite $n$. Thus, $Q^*_n$ is instead given by $(\lfloor \bold{z}^*_n \cdot n \rfloor/n, 1- (\lfloor \bold{z}^*_n \cdot n \rfloor/n)$, and
\[D_{KL}(Q^*_n \| Q) = q^*\log\left(\frac{q^*}{\lfloor \bold{z}^*_n \cdot n \rfloor/n}\right) + (1-q^*)\log\left(\frac{1-q^*}{1-\lfloor \bold{z}^*_n \cdot n \rfloor/n}\right)\]
Plugging in the given parameter values, and applying Proposition \ref{prop:Sanov}, yields the expression in the corollary.

\subsection{Examples Related to Theorem \ref{prop:Limit}} \label{app:More}

Part (a) of Theorem \ref{prop:Limit} provides a sufficient condition for the confidence set $[\underline{p}^n(a_i),\overline{p}^n(a_i)]$ to converge to certainty---$a_i$ is strictly rationalizable for type $t_i^\infty$---and Part (b) of Theorem \ref{prop:Limit} provides a necessary condition---$a_i$ is rationalizable for type $t_i^\infty$. The condition that $a_i$ is strictly rationalizable is not necessary, as I demonstrate in Section \ref{ex:NotNecessary}, and the condition that $a_i$ is rationalizable is not sufficient, as I demonstrate in Section \ref{app:RatNotSuff}.

In each of these examples, I assume (as in Section \ref{subsec:lower}) that the limiting belief $\mu^\infty$ is degenerate at a limiting parameter $\theta^\infty$, and players have common certainty of shrinking neighborhoods of this parameter. That is, for every realization $\bz_n$, players have common certainty in the event that players have first-order beliefs with support on $\mathscr{C}(\bz_n)$, where the support sets $\mathscr{C}(\bz_n)$ satisfy Assumption \ref{ass:ShrinkSupport}.

\subsubsection{Strict Rationalizability is Not Necessary} \label{ex:NotNecessary}
Consider the following complete information game
\[\begin{array}{ccc}
& a_3 & a_4 \\
a_1 & \theta,0 & \theta,0 \\
a_2 & 0,0 & 0,0\end{array}\]
and suppose that the limiting belief is degenerate at $\theta^\infty=1$. Then, the action $a_1$ is \emph{strictly dominant} for player 1 in the limiting complete information game, and also for all types with common certainty in the event that players have first-order beliefs with support on a small enough neighborhood of $\theta^\infty$. So Assumption \ref{ass:UniformConv} implies $ \lim_{n\rightarrow \infty}[\underline{p}^n(a_i),\overline{p}^n(a_i)]=\{1\}$. But action $a_1$ is not \emph{strictly rationalizable} for type $t_i^\infty$. 

\subsubsection{Rationalizability is not Sufficient}  \label{app:RatNotSuff}

I show next that rationalizability of $a_i$ for type $t_i^\infty$ is not sufficient for the analyst's confidence set for $a_i$ to converge to certainty. Section \ref{sec:exBoundary} provides a simple  example to this effect. Define $\Theta^{a_i}$ to be the set of parameter values $\theta$ such that $a_i$ is rationalizable for player $i$ in the complete information game indexed to $\theta$. If $a_i$ is on the boundary of $\Theta^{a_i}$, then common certainty of shrinking neighborhoods around $\theta^\infty$ does not guarantee rationalizability of $a_i$. More surprisingly, common certainty in arbitrarily small open sets within the \emph{interior} of $\Theta^{a_i}$ also does not guarantee rationalizability of $a_i$, and I provide an example of this in Section \ref{ex:NotSufficient}. (See also the working paper of \citet{satoru3} for a nice two-player example to this effect.)

\subsubsection{$\theta^\infty$ is on the Boundary of $\Theta^{a_i}$} \label{sec:exBoundary}

Consider the following two-player game, parametrized by $\theta \in [\underline{\theta}, \overline{\theta}]$ for some $\underline{\theta}<0<\overline{\theta}$:
\[\begin{array}{ccc}
& a & b\\
a & \theta, \theta & 0,0 \\
b & 0,0 & 1,1 
\end{array}\]
Suppose that the limiting parameter $\theta^\infty=0$, so that $a$ is rationalizable in the limiting complete information game, but not strictly rationalizable. It is straightforward to see that common certainty of shrinking neighborhoods of $\theta^\infty$ does not guarantee rationalizability of action $a$, as the type with common certainty of any $\theta'<0$ considers $a$ to be strictly dominated.

\subsubsection{$\theta^\infty$ is in the Interior of $\Theta^{a_i}$} \label{ex:NotSufficient}
But even if $\theta^\infty$ is not on the boundary of the set $\Theta^{a_i}$, it may be that common certainty of a shrinking neighborhood of $\theta^\infty$ does not guarantee rationalizability of $a_i$. Consider the following four-player game. Players 1 and 2 choose between actions in $\{a,b\}$, and player 3 chooses between matrices from $\{l,r\}$. Their payoffs are:
\begin{equation}
\label{matrix}
\begin{array}{ccc}
& a& b  \\
a& 1,1,0 & 0,0,0 \\
b& 0,0,0 & 0,0 ,0
\end{array} \quad \quad \quad \begin{array}{ccc}
& a& b  \\
a& 0,0,0 & 0,0,0 \\
b& 0,0,0 & 1,1 ,0
\end{array}
\end{equation}
\[(l) \quad \quad \quad \quad \quad \quad \quad \quad \quad \quad \,  (r)\]
A fourth player predicts whether players 1 and 2 chose matching actions or mis-matching actions. He receives a payoff of 1 if he predicts correctly (and 0 otherwise).\footnote{In more detail: player 4 chooses between $\{\text{Match},\,\, \text{Mismatch}\}$. His payoff from Match is 1 if players 1 and 2 choose the same action (both $a$ or both $b$) and 0 otherwise; his payoff from Mismatch is 1 if players 1 and 2 chose different actions ($a$ and $b$ or flipped), and 0 otherwise.} Player 4's action does not affect the payoffs of the other three players. 

Let the state space $\Theta=\mathbb{R}^{64}$ be the set of all payoff matrices given these actions, where the payoffs described above are a particular $\theta$. Match is clearly rationalizable for player 4 at $\theta$; it is also rationalizable for player 4 on a neighborhood of $\theta$ (in the Euclidean metric).\footnote{Suppose neither $l$ nor $r$ are strictly dominated for player 1; then, all actions are rationalizable for player 1-3, so Match is rationalizable for player 4. If either $l$ or $r$ is strictly dominated for player 1, then one of the following will be a rationalizable family: $\{l\} \times \{a\} \times \{a\} \times \{\mbox{Match}\}$, $\{l\} \times \{a,b\} \times \{a,b\} \times \{\mbox{Match}\}$, $\{r\} \times \{b\} \times \{b\} \times \{\mbox{Match}\}$, or $\{r\} \times \{a,b\} \times \{a,b\} \times \{\mbox{Match}\}$. Thus, Match is rationalizable for player 4. }

Nevertheless, I will show existence of a sequence of types for player 4 with common certainty in increasingly small neighborhoods of $\theta$, given which Match fails to be rationalizable. Along this sequence, player 4 believes that $a$ is uniquely rationalizable for player 1, while $b$ is uniquely rationalizable for player 2, so the action Match is strictly dominated.

Define $\theta_\varepsilon^1$ to be the following perturbation of the payoff matrix $\theta$ (with player 4's payoffs unchanged):
\begin{equation}
\label{perturbed1}
\begin{array}{ccc}
& a& b  \\
a& 1,1,0 & 0,0,0 \\
b& 0,0,0 & -\varepsilon,0 ,0
\end{array} \quad \quad \quad \begin{array}{ccc}
& a& b  \\
a& 0,0,-\varepsilon & 0,0,-\varepsilon \\
b& 0,0,-\varepsilon & 1,1,-\varepsilon 
\end{array}
\end{equation}
\[(l) \quad \quad \quad \quad \quad \quad \quad \quad \quad \quad \,  (r)\]

\noindent Let $\theta_\varepsilon^2$ correspond to the following payoff matrix (again with player 4's payoffs unchanged):
\begin{equation}
\label{perturbed2}
\begin{array}{ccc}
& a& b  \\
a& 1,1,-\varepsilon & 0,0,-\varepsilon \\
b& 0,0,-\varepsilon &0,0 ,-\varepsilon
\end{array} \quad \quad \quad \begin{array}{ccc}
& a& b  \\
a& -\varepsilon,0,0 & 0,0,0 \\
b& 0,0,0 & 1,1 ,0
\end{array}
\end{equation}
\[(l) \quad \quad \quad \quad \quad \quad \quad \quad \quad \quad \,  (r)\]
Let $\varepsilon>0$. If player 1 has common certainty in the state $\theta^1_\varepsilon$, then $a$ is his uniquely rationalizable action: $l$ strictly dominates $r$ for player 3, given which $a$ strictly dominates $b$ for player 1. By a similar argument, if player 2 has common certainty in the state $\theta^2_\varepsilon$, then $b$ is his uniquely rationalizable action. These statements hold for $\varepsilon$ arbitrarily small. Construct a sequence of types $(t^{\varepsilon_n}_4)$ for player 4, where each type $t^{\varepsilon_n}_4$ has common certainty that player 1 has common certainty in the state $\theta^1_{\varepsilon_n}$ and player 2 has common certainty in the state $\theta^2_{\varepsilon_n}$. Then, player 4 of type $t^{\varepsilon_n}_4$ has common certainty in an $\varepsilon$-neighborhood of $\theta$, but only one rationalizable action: Mismatch. Take $\varepsilon_n \rightarrow 0$ (with each $\varepsilon_n>0$) and the desired conclusion obtain: rationalizability of Match holds at $\lim_{n\rightarrow \infty} \varepsilon_n$ but fails to hold arbitrarily far out along the sequence $\varepsilon_n$. 

\subsection{Extension to Common $p$-Belief} \label{app:pbelief}

For each $q \in [0,1]$, define:
\[\underline{p}^{n,q}(a_i)= P^n \left(\left\{\bz_n\,: \, a_i\in S^\infty_i[t_i] \quad \forall t_i \in T_i^{\mathscr{B},q}\right\}\right).\]
where $q=1$ returns the definition of $\underline{p}^{n}(a_i)$ given in the main text.

\begin{proposition} \label{prop:speedRobust}  Suppose $a_i$ is strictly rationalizable for type $t_i^\infty$, and define
\[
  \delta^\infty := \sup \left\{ \delta \,: \, a_i \mbox{ is $\delta$-strictly rationalizable for type $t_i^\infty$}\right\}
\]
  noting that this quantity is strictly positive. Define
  \begin{equation} \label{eq:M}
  M:= \sup_{a,a'\in A, \theta, \theta'\in \Theta, j\in \mathcal{I}} \vert u_j(a,\theta) - u_j(a',\theta')\vert.
  \end{equation}
  Then, for every $n\geq 1$, and $q > M/(\delta^\infty+M)$,
\[
\underline{p}^{n,q}(a_i) \geq 1-\frac{2M\xi q}{\delta^\infty q - (1-q) M} \mathbb{E}\left(\sup_{\mu \in \mathcal{M}} d_P(\mu(Z^n),\mu^\infty)\right)
\]
\end{proposition}

\begin{proof}
I first demonstrate a lemma analogous to Lemma \ref{lemm:eps}.

\begin{lemma}
Suppose $a_i$ is $\delta^\infty$-strictly rationalizable for player $i$ of type $t_i^\infty$. Let $\mathscr{B}\subseteq \Delta(\Theta)$ by any set satisfying
\[\sup_{\nu \in \mathscr{B}} d_P(\nu,\mu^\infty)\leq \frac{\delta^\infty q - (1-q) M}{2M\xi q}\] where $M$ is as defined in (\ref{eq:M}) and $\xi$ is as defined in (\ref{eq:xi}). Then, $a_i$ is rationalizable for all types $t_i \in T_i^{\mathscr{B},q}$.
\end{lemma}

\begin{proof}  The proof follows along similar lines to the proof of Lemma \ref{lemm:eps}. Fix $\epsilon >0$, and consider an arbitrary set $\mathscr{B} \subseteq \{\mu^\infty\}^\epsilon$.  I will show that $a_i$ is rationalizable for all types $t_i \in T_i^{\mathscr{B},q}$ when $\epsilon$ is sufficiently small and $q$ is sufficiently large.

By assumption, action $a_i$ is $\delta^\infty$-strictly rationalizable for player $i$ of type $t_i^\infty$. This implies that there exists a family of sets $(R_j)_{j \in \mathcal{I}} \subseteq \prod_{j \in \mathcal{I}} A_j$, where $a_i \in R_i$, and for every $a_j \in R_j$ there exists a $\sigma^\infty_{-j}: \Theta \rightarrow \Delta(A_{-j})$ satisfying
 \[\supp \sigma^\infty_{-j}(\theta) \subseteq R_{-j} \quad \forall \theta \in \Theta\]
 and  
\begin{equation} \label{eq:Inftyq}
\int_{\Theta} u_j(a_j, \sigma^\infty_{-j}(\theta), \theta) d\mu^\infty  - \int_{\Theta} u_j(a'_j, \sigma^\infty_{-j}(\theta), \theta) d\mu^\infty  \geq \delta^\infty \quad \forall a_j'\neq a_j
\end{equation}

Partition the set of types $T_j^{\mathscr{B},q}$ into those types whose first-order beliefs belong to $\mathscr{B}$
\[T^1_j := \left\{ t_j \in T_j^{\mathscr{B},q} \mid t_j^1\in \mathscr{B}\right\} \quad \forall j \in \mathcal{I}\]
 and all remaining types $\overline{T}^1_j := T_j^{\mathscr{B},q} \backslash T^1_j$. By construction, every type in $T_j^{\mathscr{B},q}$ assigns probability at least $q$ to $T_{-j}^1$. I will now show that there exists a family of sets $(V_j[t_j])_{j \in \mathcal{I},t_j \in T_j^{\mathscr{B},q}}$ with the property that for each $k\geq 1$, player $j$, type $t_j \in T_j^{\mathscr{B},q}$, action $a_j \in R_j$, and mixed strategy $\alpha_j \in \Delta(A_j \backslash \{a_j\})$, there exists a measurable $\sigma_{-j}: \Theta \times T_{-j}^{\mathscr{B},q} \rightarrow \Delta(A_{-j})$ with 
\begin{enumerate}
\item[(1)] $\supp \sigma_{-j}(\theta, t_{-j}) \subseteq V_{-j}[t_{-j}] \quad \forall (\theta, t_{-j}) \in \Theta \times T^{\mathscr{B},q}_{-j}$
\item[(2)] $\int_{\Theta \times T^{\mathscr{B},q}_{-j}} [u_j(a_j, \sigma_{-j}(\theta, t_{-j}), \theta) - u_j(\alpha_j, \sigma_{-j}(\theta,t_{-j}),\theta] t_j[d\theta \times dt_{-j}] \geq 0$
\item[(3)] $V_j[t_j]=R_j$ for every player $j$ and type $t_j \in T^1_j$. 
\end{enumerate}
Since $a_i \in R_i$ by design, it follows from Proposition \ref{prop:CDFX} that for any type $t_i \in T_i^{\mathscr{B},q}$, the action $a_i \in S^k_i[t_i]$ for every $k$, and hence $a_i \in S^\infty_i[t_i]$, as desired.

Fix an arbitrary player $j$, $a_j \in R_j$, type $t_j \in T_j^{\mathscr{B},q}$, and $\alpha_j \in \Delta(A_j \backslash \{a_j\})$. Define $a_{-j}: \Theta \rightarrow A_{-j}$ to satisfy
\[a_{-j}(\theta) \in \argmax_{a_{-j}\in R_{-j}} (u_j(a_j,a_{-j},\theta)- u_j(\alpha_j,a_{-j},\theta)) \quad \forall \theta \in \Theta\]
Further define 
\begin{equation} \label{eq:h}
h(\theta) := u_j(a_j, a_{-j}(\theta), \theta) - u_j(\alpha_j, a_{-j}(\theta),\theta) \quad \forall \theta \in \Theta.
\end{equation}
Let $\sigma_{-j}: \Theta \times T_{-j}^{\mathscr{B},q} \rightarrow \Delta(A_{-j})$ be any conjecture with the property that $\sigma_{-j}(\theta, t_{-j})$ is a point mass at $a_{-i}(\theta)$ for every $(\theta, t_{-j}) \in \Theta \times T_{-j}^{\mathscr{B},1}$. The conjectures $\sigma_{-j}(\theta, t_{-j})$ for $(\theta, t_{-j})\notin \Theta \times T_{-j}^{\mathscr{B},1}$ are not explicitly specified. By definition,
\[\supp \sigma_{-j}(\theta,t_{-j}) \subseteq R_{-j} \quad \forall (\theta, t_{-j}) \in \Theta \times T_{-j}^{\mathscr{B},1}.\]
Then

\begin{align*}
\int_{\Theta \times T^1_{-j}}  &  u_j(a_j, \sigma_{-j}(\theta, t_{-j}), \theta) t_j[d\theta \times dt_{-j}] - \int_{\Theta \times T^1_{-j}} u_j(\alpha_j, \sigma_{-j}(\theta, t_{-j}), \theta) t_j[d\theta \times dt_{-j}] \\
& = \left(\int_{\Theta \times T^1_{-j}}    u_j(a_j, \sigma_{-j}(\theta, t_{-j}), \theta) t_j[d\theta \times dt_{-j}] + \int_{\Theta \times \overline{T}^1_{-j}}  u_j(a_j, \sigma_{-j}(\theta, t_{-j}), \theta) t_j[d\theta \times dt_{-j}]\right) -  \\
& \quad \quad \left(\int_{\Theta \times T^1_{-j}}    u_j(\alpha_j, \sigma_{-j}(\theta, t_{-j}), \theta) t_j[d\theta \times dt_{-j}] + \int_{\Theta \times \overline{T}^1_{-j}}  u_j(\alpha_j, \sigma_{-j}(\theta, t_{-j}), \theta) t_j[d\theta \times dt_{-j}]\right)\\
& = \left(\int_{\Theta \times T^1_{-j}}    u_j(a_j, a_{-j}(\theta), \theta) t_j[d\theta \times dt_{-j}] -  \int_{\Theta \times T^1_{-j}}    u_j(\alpha_j, a_{-j}(\theta), \theta) t_j[d\theta \times dt_{-j}] \right) + \\
& \quad \quad \left(\int_{\Theta \times \overline{T}^1_{-j}}  u_j(a_j, \sigma_{-j}(\theta,t_{-j}), \theta) t_j[d\theta \times dt_{-j}] - \int_{\Theta \times \overline{T}^1_{-j}}  u_j(\alpha_j, \sigma_{-j}(\theta, t_{-j}), \theta) t_j[d\theta \times dt_{-j}]\right)\\
& =  \int_{\Theta\times T_{-j}^1} h(\theta) t_j[d\theta\times dt_{-j}] + \\
& \quad \quad \left(\int_{\Theta \times \overline{T}^1_{-j}}  u_j(a_j, \sigma_{-j}(\theta,t_{-j}), \theta) t_j[d\theta \times dt_{-j}] - \int_{\Theta \times \overline{T}^1_{-j}}  u_j(\alpha_j, \sigma_{-j}(\theta, t_{-j}), \theta) t_j[d\theta \times dt_{-j}]\right)\\
& \geq \int_{\Theta\times T_{-j}^1} h(\theta) t_j[d\theta\times dt_{-j}] - M \int_{\Theta \times \overline{T}^1_{-j}} t_j[d\theta \times dt_{-j}] 
\end{align*}
where the final inequality follows from the definitions of $h$ (as given in (\ref{eq:h})) and $M$ (as given in (\ref{eq:M})). 

In the proof of Lemma \ref{lemm:eps}, we showed that the inequality in (\ref{eq:Inftyq}) implies $\int_\Theta h(\theta)t_j^1[d\theta] \geq \delta^\infty - 2M\xi \epsilon$. Since moreover $t_j$ assigns probability at least $p$ to the set $T^1_{-j}$, we can further bound
\[ \int_{\Theta\times T_{-j}^{\mathscr{B},1}} h(\theta) t_j[d\theta\times dt_{-j}] + M \int_{\Theta \times \overline{T}^1_{-j}} t_j[d\theta \times dt_{-j}]  \geq q(\delta^\infty - 2M\xi \epsilon) - (1-q)M\]
Thus, $a_j$ is a best reply for type $t_j$ so long as 
\[q(\delta^\infty - 2M\xi \epsilon) - (1-q)M \geq 0\]
which simplifies to
$
\epsilon \leq \frac{\delta^\infty q - (1-q) M}{2M\xi q}
$
This bound holds across all players $j$ and actions $a_j \in T_j^1$.
\end{proof}
Thus
\begin{align*}
\underline{p}^n(a_i) & \geq  P^n \left(\left\{ \bz_n \, : \, \sup_{\mu \in \mathcal{M}} d_P(\mu(\bz_n),\mu^\infty) \leq \frac{\delta^\infty q - (1-q) M}{2M\xi q} \right\} \right) \\
& \geq 1- \frac{2M\xi p}{\delta^\infty q - (1-q) M}  \mathbb{E} \left( \sup_{\mu \in \mathcal{M}} d_P(\mu(Z^n),\mu^\infty)  \right) 
\end{align*}
using Markov's inequality in the final line.
\end{proof}

\end{document}